\documentclass[]{article}
\usepackage{natbib}
\usepackage{graphics}
\usepackage{epsfig}
\usepackage{authblk}
\usepackage{amssymb}
\usepackage{graphicx}
\usepackage{bm}
\usepackage{amsthm}

\usepackage[nolists]{endfloat}
\usepackage{graphicx}
\usepackage[centertags]{amsmath}
\usepackage[ansinew]{inputenc}
\usepackage[spanish,english]{babel}
\usepackage{amsfonts}
\usepackage{latexsym}
\usepackage{array}
\usepackage{epsfig}
\usepackage{color}
\usepackage{ulem}
\usepackage{hyperref}
\usepackage{cleveref}

\newcommand{\beq}{\begin{equation}}
\newcommand{\eeq}{\end{equation}}
\newcommand{\bed}{\begin{displaymath}}
\newcommand{\eed}{\end{displaymath}}
\newcommand{\bbeta}{\vec{\beta}}
\newcommand{\ppsi}{\vec{\psi}}

\newcommand{\btheta}{\vec\theta}

\newcommand \bdelta {\vec\delta}

\newcommand \bUpsilon{\vec\Upsilon}
\newcommand \bPsi{\vec\Psi}
\newcommand \bPhi{\vec\Phi}

\newtheorem{prop}{Proposition}

\def \bA {{\mathbf A}}
\def \ba {{\mathbf a}}
\def \bB {{\mathbf B}}
\def \bC {{\mathbf C}}
\def \bc {{\mathbf c}}

\def \bd{{\mathbf d}}
\def \bb {{\mathbf b}}
\def \be {{\mathbf e}}
\def \bf {{\mathbf f}}

\def \bK {{\mathbf K}}
\def \bH {{\mathbf H}}

\def \bI {{\mathbf I}}
\def \bL {{\mathbf L}}
\def \bl {{\mathbf l}}
\def \br {{\mathbf r}}

\def \bR {{\mathbf R}}

\def \bt {{\mathbf t}}
\def \bX {{\mathbf X}}
\def \bx {{\mathbf x}}
\def \by {{\mathbf y}}
\def \bu {{\mathbf u}}

\def \bV {{\mathbf V}}
\def \bU {{\mathbf U}}
\def \bT {{\mathbf T}}
\def \bZ {{\mathbf Z}}
\def \bz {{\mathbf z}}
\def \bW {{\mathbf W}}

\def \bzero {{\mathbf 0}}
\def \bSigma {{\mathbf \Sigma}}

\def \bLambda {{\mathbf \Lambda}}

\def \bOmega{{\mathbf \Omega}}

\def \uno {\mathbf 1}

\title{Small Area Estimation with Linked Data}
\author[$1$] {R. Chambers}
\author[$2$] {E. Fabrizi}
\author[$3$] {N. Salvati}

\affil[$1$] {{\small National Institute for Applied Statistics Research Australia, School of Mathematics and Applied Statistics, University of Wollongong }}
\affil[$2$] {{\small Dipartimento di Scienze Economiche e Sociali, Universit\`a Cattolica del Sacro Cuore di Piacenza }}
\affil[$3$] {{\small Dipartimento di Economia e Management, Universit\`a di Pisa}}

\date{}

\begin{document}
	
	%
	%
	%
	%
	
	
	\maketitle
	
	\begin{abstract}

		In Small Area Estimation data linkage can be used to combine values of the variable of interest from a national survey with values of auxiliary variables obtained from another source like a population register. Linkage errors can induce bias when fitting regression models; moreover, they can create non-representative outliers in the linked data in addition to the presence of potential representative outliers. In this paper we adopt a secondary analyst's point view, assuming limited information is available on the linkage process, and we develop small area estimators based on linear mixed and linear M-quantile models to accommodate linked data containing a mix of both types of outliers. We illustrate the properties of these small area estimators, as well as estimators of their mean squared error, by means of model-based and design-based simulation experiments. These experiments show that the proposed predictors can lead to more efficient estimators when there is linkage error. Furthermore, the proposed mean-squared error estimation methods appear to perform well.
		
	\end{abstract}
	\vspace{9pt}
	
	\section{Introduction}\label{intro}
	Estimates of finite population parameters are often needed for subsets (domains) of the population, defined either by geographical disaggregation (areas) or by other classification criteria (e.g. region by gender by age class). When the domain-specific portion of the available data is so small that standard estimators are unacceptably imprecise for most of the domains, we have a small area estimation (SAE) problem. See \cite{Pfe13} and \cite{RaoMol15} for general introductions to the topic. From now on we refer to the domains of interest as areas.

	Small area estimation methods complement available data, typically from a large population survey, with area specific auxiliary information. A standard setting is where it is reasonable to assume that the value $y_{ij}$ of the target variable for unit $j$ in area $i$ is related to a known vector of covariates $\bx_{ij}$ by means of a regression model. These $\bx_{ij}$ values, assumed to be known for both the survey sample and the rest of the population, are then used to predict the area parameter of interest.
	
	Data integration is fast becoming an intrinsic part of Official Statistics, in large part due to the increasing availability of administrative registers and other population data sources. Here we focus on the situation where the $y_{ij}$ values are measured in a sample survey but where the $\bx_{ij}$ are not measured in the same survey. Instead, these values are extracted from a population register and then linked to the sampled units. An illustrative example is where a limited number of variables are collected using the survey questionnaire and the sample records are then linked to unit level information stored in a separate population register in order to complete a dataset for small area estimation. If an error-free unique identifier exists in both the survey record and the population register, this linkage can be deterministic. However, in many cases such an identifier is not error free or does not exist, in which case we need to allow for record linkage errors.
	
	Due to its growing importance, record linkage has attracted considerable scientific interest. Broadly speaking, we can identify two main literature streams: the first concerned with how to link records when an error-free unique identifier is missing; the second focused on how to adjust statistical methods so that they are appropriate for the analysis of linked data containing linkage errors. For recent reviews of the first literature stream see \cite{Winkler2009}, \cite{Winkler2014} and \cite{Lahiri18}.
	
	It is widely recognized that overlooking linkage errors when analysing linked data can lead to biased estimates even if most records are correctly linked. Bias correction methods when fitting linear regression models to linked data are discussed in \cite{Scheurenwinkler1993}, \cite{Scheurenwinkler1997}, \cite{lahirilarsen2005}, \cite{OSR2009}, \cite{Kimchambers2012} and \citet{Lahiri18}. The impact of linkage errors on linear mixed models, which are often used in small area estimation, has received comparatively less attention \citep{Samartchambers2014}. More specifically, we are aware of just one other article \citep{Briscolini2018} where linear mixed models are used with linked data for small area estimation.
	
	However, there is another aspect to linkage errors that seems to have attracted much less attention. This is when linkage errors generate artificial outliers in the linked data set. Let $y_{ij}^\star$ denote the linked value corresponding to $y_{ij}$. Such an outlier can then be generated when there is linkage error, and the residual associated with the correctly linked pair $(y_{ij},\bx_{ij})$ is small, but the residual associated with the incorrectly linked pair $(y_{ij}^\star,\bx_{ij})$ is large. This can happen when the variables used in the matching process (such as names, addresses, identification codes) are independent of those used as regressors. In Figure \ref{pop_example} we illustrate this phenomenon using a synthetic population, which is described in more detail in Section \ref{mod:based}. In the upper panel a scatterplot shows the strong linear relationship between $x$ and $y$ when there are no linkage errors; in the lower panel the same relationship is shown when linkage errors occur at an overall rate of 28.5 per cent. In the lower panel incorrectly matched pairs $(y_{ij}^\star,\bx_{ij})$ are shown as open circles. Outlying residuals are evident.
	
	\begin{figure}[h]
		\centering    
		\makebox{\includegraphics[scale = 0.50]{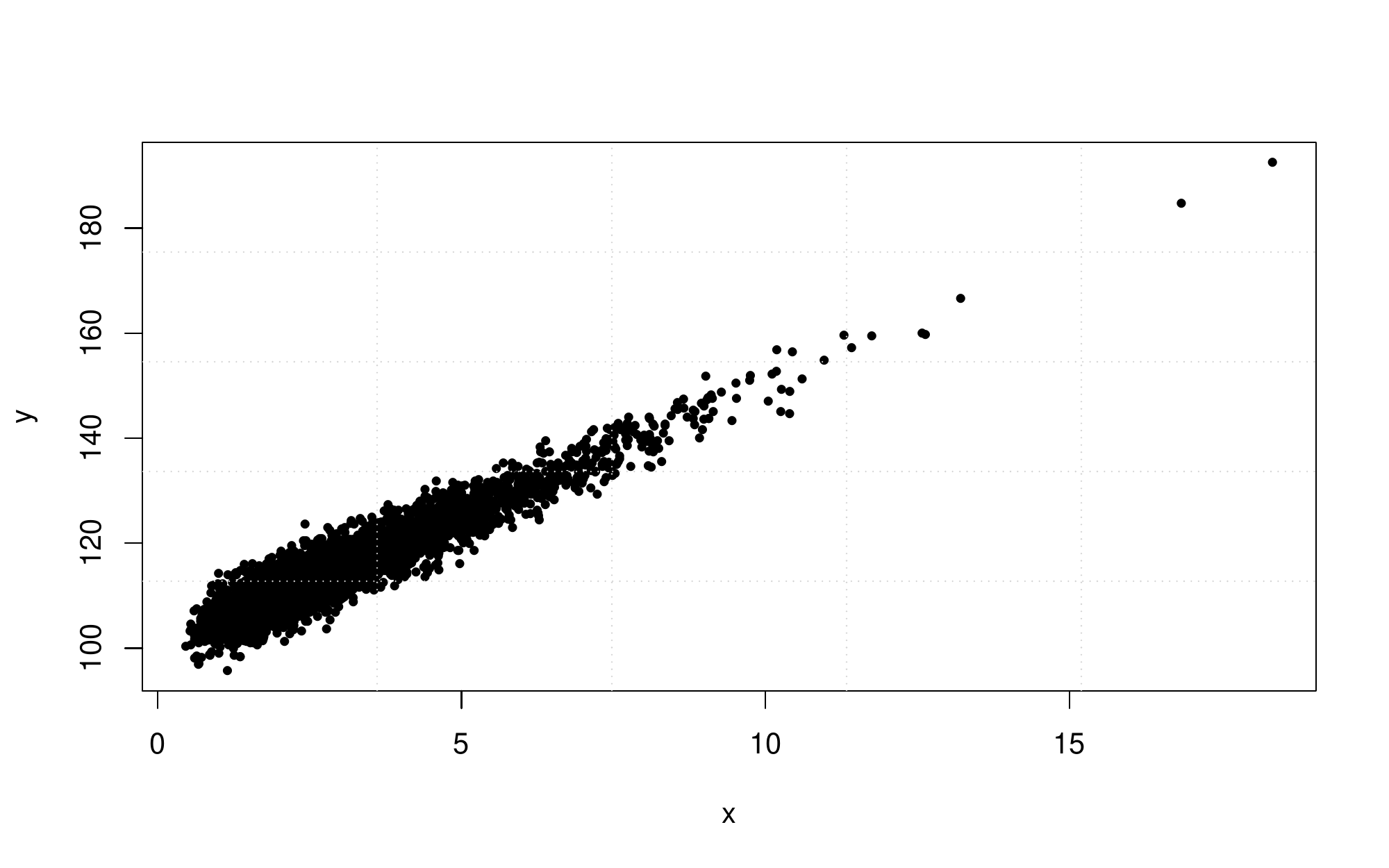}}\\
		\makebox{\includegraphics[scale = 0.50]{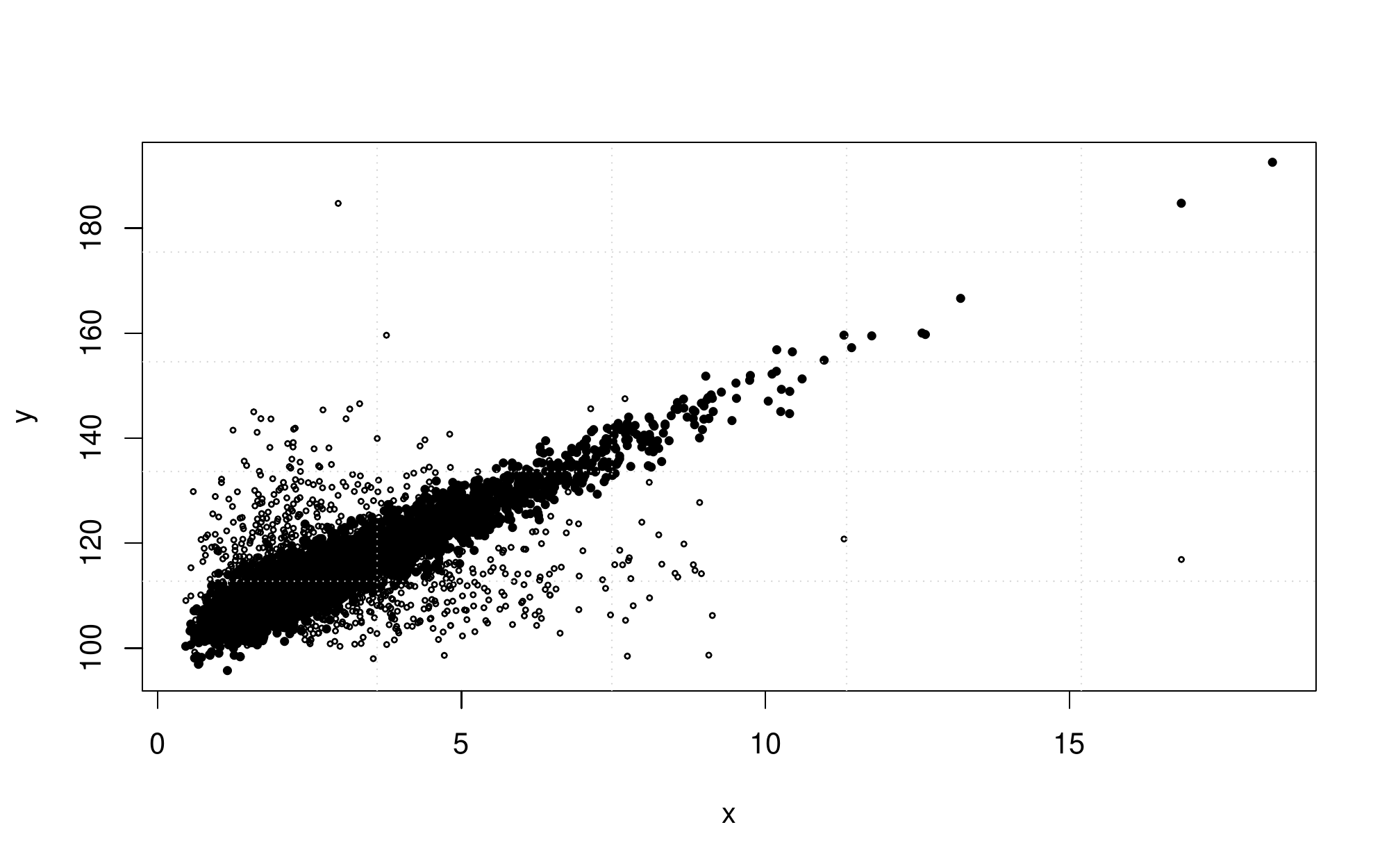}} 
		\caption{\label{pop_example} Scatterplots of two synthetic populations. In the upper panel the relationship between $x$ and $y$ when there are no linkage errors (filled circle). In the lower panel the same relationship when wrongly matched pair $(y_{ij}^\star,\bx_{ij})$ occurring (circle).}
	\end{figure}
	
	\cite{Cha86} first distinguished between representative and non-representative outliers in a survey sample. Using this distinction, these artificial outliers are non-representative, and so are fundamentally different from outliers associated with the correctly linked population units, which are representative. The problem is that is not possible to tell \textit{a priori} whether an outlier is induced by linkage error (and so is non-representative) or is representative. In particular, outliers due to linkage errors can violate the assumptions underpinning non-robust estimation methods, as well as cause bias problems for robust projective \citep{Cha14} estimation methods since downweighting linkage error-induced outliers does not generally rid the sample data of linkage errors. This problem is illustrated in the scatterplots shown in Figure \ref{pop_example1}, which are for an outlier-prone version of the same synthetic population underpinning Figure  \ref{pop_example}. Here the upper panel shows the relationship between $x$ and $y$ when there are representative outliers (denoted by filled triangles, point-up), whereas the scatterplot in the lower panel shows the same population when linkage errors leads to wrongly matched pairs $(y_{ij}^\star,\bx_{ij})$ (denoted by open circles). It is difficult to distinguish the residuals due to representative outliers from those due to linkage errors in this lower panel. 
	
	\begin{figure}[h]
		\centering    
		\makebox{\includegraphics[scale = 0.50]{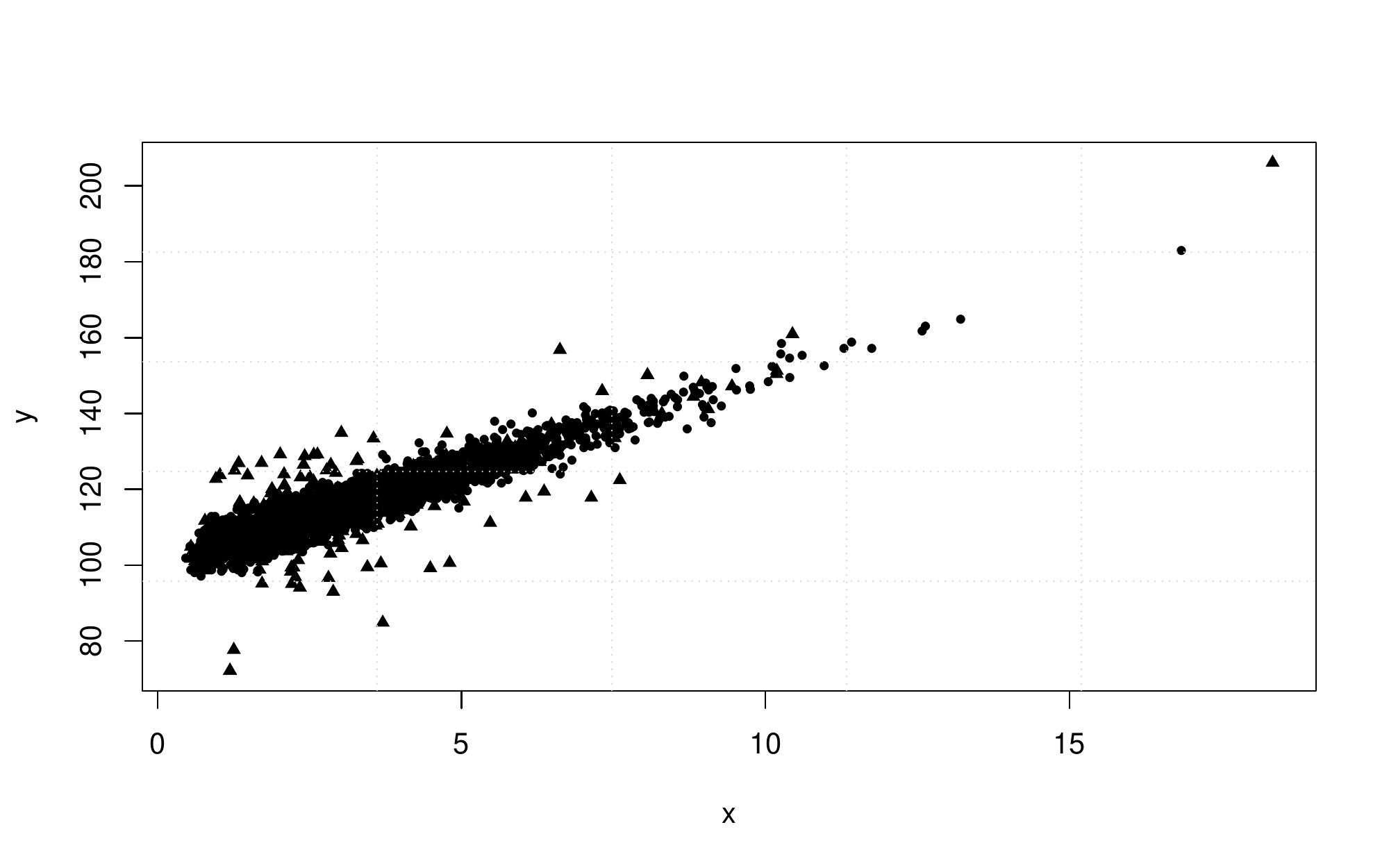}}\\
		\makebox{\includegraphics[scale = 0.50]{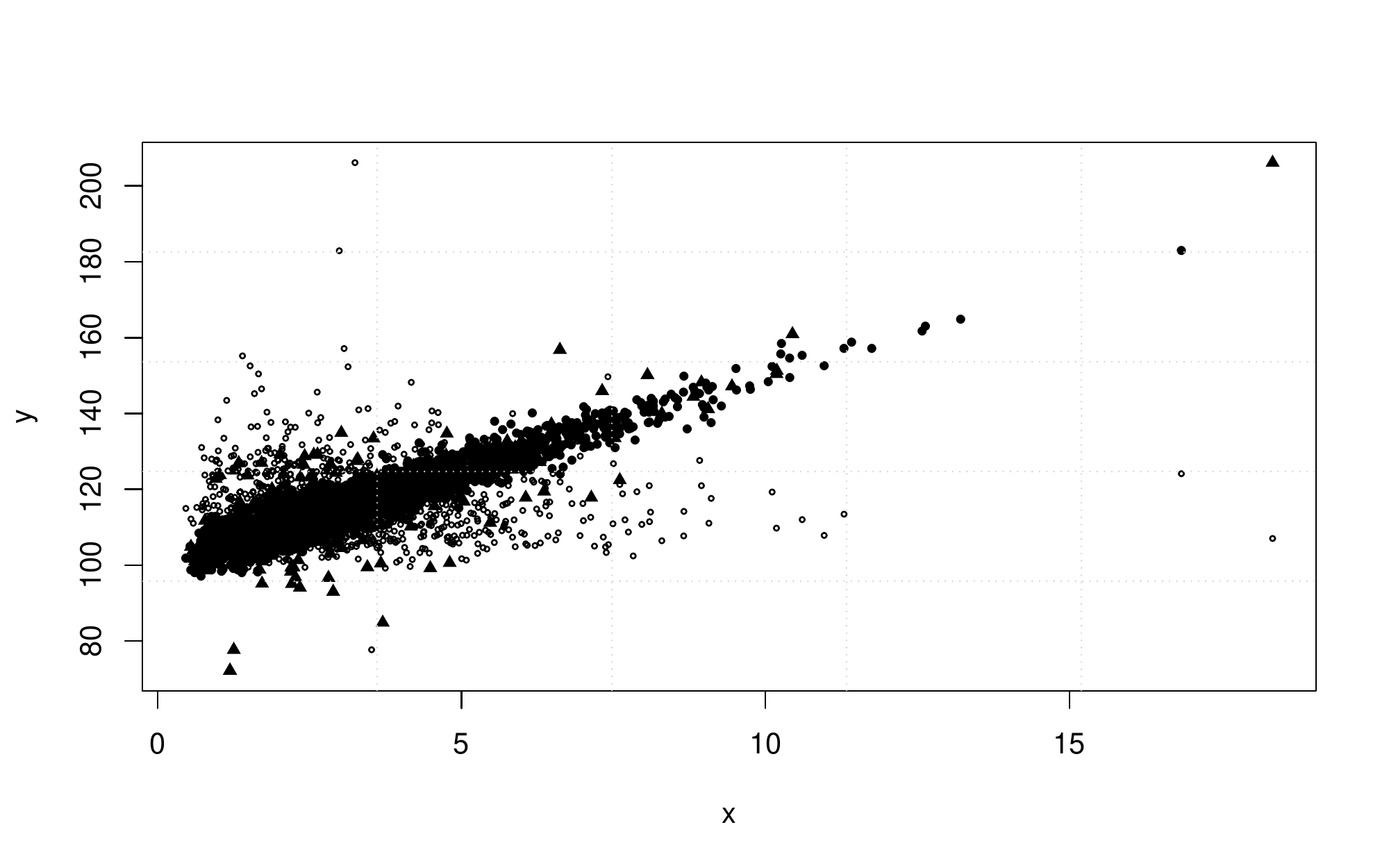}} 
		\caption{\label{pop_example1} Scatterplots of two synthetic populations. In the upper panel is represented the relationship between $x$ and $y$ when there are representative outliers (filled triangle point-up). In the lower panel the same relationship when also wrongly matched pair $(y_{ij}^\star,\bx_{ij})$ occurring (circle).}
	\end{figure}

	Two questions immediately arise when one considers Figure \ref{pop_example1}. The first is whether well-known outlier robust methods for small area estimation can adequately deal with the mix of representative and non-representative outliers that can potentially occur in a linked data situation. The second concerns appropriate modifications to these methods to allow for linkage errors. We consider both in what follows. 
	
	In this paper  we contribute to the literature by extending some popular small area estimators based on linear models, including the Empirical Best Linear Unbiased Predictor \citep[EBLUP,][]{Bat88}, its robust version \citep[REBLUP,][]{Sin09} and the M-quantile-based predictor \citep{Cha06, Tza10}, to non-deterministically linked data. We also propose analytical MSE estimators for the modified estimators that we introduce. In our  development we adopt a secondary analyst viewpoint, that is we assume that the researcher producing the small area estimates does not have access to all the information used in the linkage process. Instead, we assume that he/she has access to the linked data set (including area indicators, which are assumed to be without error) and is also provided with minimal information regarding the linkage quality. We characterise this minimal information via a simple exchangeable linkage error structure within specified poststrata, which are referred to as blocks below. In contrast, \citet{Briscolini2018} take a primary analyst viewpoint and so assume that a richer set of information on the linkage process is available; although not directly concerned with small area estimation, the same comment applies to \citet{Lahiri18}.
	
	The paper is organized as follows. Section \ref{sec:ass} is devoted to setting out the theoretical background and the assumptions of the linkage error model which is then used to extend the small area predictors.
	In Sections \ref{sec:eblup_star}, \ref{sec:REBLUP} and \ref{sec:MQ} we introduce the proposed extensions to the EBLUP, REBLUP and M-quantile-based predictors under linkage error and their corresponding MSE estimators. In Sections \ref{mod:based} and \ref{sec:db} the performances of these newly proposed predictors are empirically assessed, both in terms of point estimation performance as
	well as in terms of MSE estimation, by means of a model-based simulation study that considers a number of different scenarios as well as by a design-based simulation. Finally, in Section \ref{sec:conclusions} we summarize our main findings,  and provide directions for future research.

	\section{Background and assumptions}\label{sec:ass}
	
	For simplicity of exposition, we restrict our development to the case of two registers, one containing values $\bx_{ij}$, and the other containing values $y_{ik}$. Extension to the case of multiple linked registers can be carried out along the same lines set out in \citet{KimCha2015}. We assume that both registers contain no duplicates, correspond to the same finite population $U$ of size $N$, and include a set of unit identifiers (linking variables) that are measured without error and are used for matching purposes. In principle, both registers can be matched on a one to one basis. However the linking variables are not unique identifiers, and linkage errors are possible.
	
	We model linkage errors from a secondary analyst point of view. In particular, we assume that the analyst has access to identifiers that allow each register to be partitioned into $Q$ non-overlapping subsets or blocks such that linkage errors are homogeneous within a block and possibly heterogeneous between blocks. The block identifiers themselves are assumed to depend on one or more linking variables, and as a consequence linkage errors are possible within blocks but not across blocks. We also assume that $U$ can be partitioned into $D$ non-overlapping areas or domains, and that the linking is carried out within an area, so two population units from different areas cannot be erroneously matched.  Cross-classifying $U$ by the area and block indicators, we then define $U_{iq}$ to be the subset of $N_{iq}$ population units that make up the segment of area $i$ nested within block $q$, with $i = 1, \dots, D$ and $q = 1, \dots, Q$. We use {\color{black}$\bx_{iqj}$ and $\by_{iqk}$} to denote individual population values from the two registers associated with this $iq$ cell.
	
	By definition, secondary analysts, i.e., analysts who are {\it not} involved in the data linking process, do not have access to the detailed information used in matching. The data available to such an analyst are therefore somewhat limited, though some information about the accuracy of the linkage process may be available. As we shall see, access to this linkage {\it paradata} is necessary before one can account for linkage errors in analysis. Furthermore, information about how the sample containing the linked data was obtained is also necessary. Since we focus on unit level modelling for small area estimation, we shall assume that the sample was drawn from register $X$ containing the {\color{black}$\bx_{iqj}$} and that this sample was then linked to register $Y$ containing the {\color{black}$y_{iqk}$}; for example $X$ could be a Census register and $Y$ could be a tax register. A related scenario is where $Y$ corresponds to a frame with contact information that is used to select a sample of units in $U$, with this sample then linked to a register $X$. What is important in both cases, however, is that the register (or frame) $Y$, with its area and block identifiers, is not available. What is available are the linked sample data (including area and block identifiers) plus, at a minimum, area by block-specific summary data from $X$.
	
	Let $s_{iq}$ denote the set corresponding to the $n_{iq}$ population indexes of the sample units in small area $i$ and block $q$, with  $n = \sum_{i = 1}^{D} \sum_{q=1}^{Q}n_{iq}$. The set containing the $N_{iq}-n_{iq}$ indices of the non-sampled units in small area $i$ and block $q$ is denoted by $r_{iq}$. For ease of notation, we assume that all areas are sampled, noting that non-sampled areas are easily accommodated. We also assume that sampling is non-informative for the small area distribution of {response variable given the covariates}, allowing us to use population level models with the sample data. Let $\by_{iq}$ denote the $N_{iq}$ vector of values for {\color{black}$y_{iqk}$} in $U_{iq}$, with $\bX_{iq}$ denoting the $N_{iq}\times p$ matrix with rows defined by the {\color{black}$\bx_{iqj}$} values of the corresponding population units. The sample components of these quantities are then denoted $\by_{siq}$ and $\bX_{siq}$ respectively. Unless the linkage is perfect, $\by_{iq}$ is unknown. Instead, what we observe is a a sample from the vector $\by_{iq}^\star$ containing the linked values $y_{iqj}^\star$ generated by the linkage process. Given that both registers can be matched on a one to one basis, we characterise the relationship between $\by_{iq}$ and $\by_{iq}^\star$ via a latent random permutation matrix $\bA_{iq} = [a_{jk}^{iq}]$ of order $N_{iq}$. That is, we put
	\begin{eqnarray}
	\label{lem} \by_{iq}^\star=\bA_{iq}\by_{iq}. 
	\end{eqnarray}
	\noindent The distribution of linkage errors in $U_{iq}$ is then determined by the distribution of $\bA_{iq}$. In general, this distribution depends on all the information used in the linking process, which, as has already been noted, is typically unavailable. However, a minimum amount of information about the accuracy of the linkage may be available, in which case a secondary analyst should be able to model the linkage errors within $U_{iq}$ via a simple exchangeable linkage error (ELE) specification:
	\begin{eqnarray}
	\label{ele1} Pr(\text{correct linkage})&=&Pr(a_{jj}^{iq}=1)=\lambda_q\\
	\label{ele2} Pr(\text{incorrect linkage})&=&Pr(a_{jk}^{iq}=1)=\gamma_q=\frac {1-\lambda_q}{N_{iq}-1},
	\end{eqnarray}
	with $j,k=1\dots,N_{iq}$. Note that although units from two different areas within a block can never be (incorrectly) matched, we assume that probabilities of correct linkage $\lambda_q$ are the same within all area segments contained within that block. Furthermore, given that we know which block is being referred to, this probability is the same irrespective of the values of $\bX_{iq}$ and $\by_{iq}$. That is, linkage is non-informative for the distribution of $\by_{iq}$ given $\bX_{iq}$. As a consequence 
	\[
	E_{\bA}(\bA_{iq}|\bX_{iq})=\bT_{iq}=(\lambda_q-\gamma_q)\mathbf{I}_{N_{iq}}+\gamma_q\mathbf{1}_{N_{iq}}\mathbf{1}_{N_{iq}}^{\prime},
	\]
	where $\mathbf{I}_{N_{iq}}$ denotes the identity matrix of order $N_{iq}$, $\mathbf{1}_{N_{iq}}$ denotes a vector of ones of length $N_{iq}$ and $E_{\bA}(.)$ denotes expectation with respect to the linkage error model. It immediately follows that we can write,
	\begin{equation}\label{noninfo} 
	E_M E_{\bA}(\bA_{iq}\by_{iq}|\bX_{iq})=E_{\bA}(\bA_{iq}|\bX_{iq})E_M(\by_{iq}|\bX_{iq})=\bT_{iq}E_M(\by_{iq}|\bX_{iq}). 
	\end{equation}
	Here $E_M(.)$ denotes expectation with respect to the model for $\by_{iq}$ given $\bX_{iq}$.
	
	Without loss of generality, we partition the matrix $\bA_{iq}$ as 
	\begin{displaymath}
	\bA_{iq}=\left[\begin{array}{c}\bA_{siq} \\  \bA_{riq} \end{array}\right],
	\end{displaymath}
	where $\bA_{siq}$ is a $n_{siq} \times N_{iq}$ matrix and $\bA_{riq}$ is $(N_{iq}-n_{iq}) \times N_{iq}$ matrix. These matrices contain the rows of $\bA_{iq}$ corresponding to sampled and non-sampled units, respectively. Then $\by^\star_{siq}$, $\bX_{iq}$ are observed, i.e. available to the analyst, while $\by_{siq}$ is not observed, where
	\begin{equation}
	\label{samplelem}
	\by_{siq}^\star=\bA_{siq}\by_{iq}.
	\end{equation}
	As noted above, the matrix $\bA_{siq}$ is not observable, but  under the ELE assumptions (\ref{ele1}) and (\ref{ele2}) we have that
	\begin{equation}\label{tsiq}
	E_{\bA}(\bA_{siq}|\bX_{iq})=\bT_{siq}=
	\begin{array}{c|c}
	\big[(\lambda_q-\gamma_q) \bI_{n_{iq}}~~ &~~\mathbf{0}_{r_{iq}} \big] + \gamma_{q} \uno_{n_{iq}} \uno_{N_{iq}}^\prime,
	\end{array}
	\end{equation}
	where $\mathbf{0}_{r_{iq}}$ is a $n_{iq} \times (N_{iq}-n_{iq})$ matrix of zeroes.
	We shall assume that the values of $\lambda_q$ are known or can be accurately estimated from the information in the linkage paradata. Unless stated otherwise, from now on we therefore condition our analysis on these values of $\lambda_q$ (and hence on $\bT_{siq}$). We briefly discuss how to account for the uncertainty due to estimation of this parameter in Section \ref{sec:conclusions}.
	More generally, as far the analyses of this paper are concerned, we shall assume that the sampled rows and the column means of $\bX_{iq}$ are known. As a consequence, the matrix
	\begin{equation}
	\bX_{siq}^\star=E_{\bA}\big(\bA_{siq}\bX_{iq} | \bX_{iq} \big)=\bT_{siq}\bX_{iq}=
	\big\{(\lambda_q-\gamma_q) \bX_{siq} + \gamma_{q} N_{iq}  \uno_{n_{iq}}\bar{\bx}_{iq}\big\},
	\end{equation} 
	will be treated as known when conditioning on $\lambda_q$.
	
	%
	%
	\section{Linear mixed models for small area estimation with linked data}\label{sec:eblup_star}
	Linear mixed models for population unit data are widely used for SAE. These models include area-specific random effects that are used to characterise area level heterogeneity in the model residuals. See  \citet{Bat88} for an early example of their application. A general specification for a unit level linear mixed model used in SAE is
	\begin{equation}
	\by=\bX\bbeta+\bZ\bu+\be,
	\label{EBLUP_model}
	\end{equation}
	where $\by$ and $\bX$ denote the population level vector of response variable and matrix of covariates, respectively; $\bu=(\bu_1^{\prime},\cdots,\bu_D^{\prime})$ is a vector of dimension $Dm$ made up of $D$ independent realizations $\{\bu_i; i=1,\cdots,D\}$ of a $m$-dimensional random area effect with $\bu \sim N(\bzero, \bSigma_u)$ and $\be \sim N(\bzero, \bSigma_e)$ is the $N\times 1$ vector of individual errors. Since the random effects $\bu$ and the individual errors $\be$ are independent, the covariance matrix of $\by$ is $\bSigma=\bSigma_e+\bZ\bSigma_u\bZ^\prime$. Here $D$ is the total number of small areas that make up the population and $m$ is the dimension of $\bz_{ij}$ so that $\bZ$ is an $N \times Dm$ matrix of fixed known constants that do not vary within an area. We assume that the covariance matrices $ \bSigma_u$ and $ \bSigma_e$ are defined in terms of a lower dimensional set of parameters $\bdelta=(\delta_1, \cdots, \delta_K)$, which are typically referred to as the variance components of model \eqref{EBLUP_model}, whereas the vector $\bbeta$ stands for the $p\times 1$ vector of regression coefficients. Provided that it is reasonable to assume that the distribution of the unit level residuals in $\be$ remains the same from block to block, then at the $U_{iq}$ sub-population level, (\ref{EBLUP_model}) can be written as:
	\begin{equation}\label{Uiqlevel}
	\by_{iq}=\bX_{iq}\bbeta+\bZ_{iq}\bu_{i}+\mathbf{e}_{iq},
	\end{equation}
	where $\bZ_{iq}$ is the $N_{iq} \times m$ incidence matrix defined by the rows of $\bZ$ corresponding to area $i$ units in block $q$.
	
	Now suppose that linked data are used to define the sample values of the response variable within sub-population $U_{iq}$. We model these data by assuming that they are obtained by non-informative sampling (given the covariates defining $\bX_{iq}$) from the outcome of a hypothetical one to one and complete linking of the records defining $\bX_{iq}$ to the records defining $\by_{iq}$. That is, the linkage error process within sub-population $U_{iq}$ can be characterised by the permutation matrix $\bA_{iq}$, and (\ref{lem}) applies. Since the sample selection process is non-informative, we can use (\ref{samplelem}) to write down a model for the linked sample values $\by^\star_{siq}$ that takes into account the linkage error process:
	\begin{equation}\label{sample_mod}
	\by^\star_{siq}=\bA_{siq}\by_{iq}=\bA_{siq}\bX_{iq}\bbeta+\bZ_{siq}\bu_{i}+\mathbf{e}^\star_{siq},
	\end{equation}
	where $E_M(\mathbf{e}^\star_{siq})=\mathbf{0}$ and $V_M(\mathbf{e}^\star_{siq})=\bSigma_{seiq}$, i.e. the sampled rows of the area $i$ by block $q$ component of $\bSigma_e$. Also, since the values contained in the columns of $\bZ_{iq}$ do not change within an area, and because we assume that matching across areas is impossible, it follows that $\bA_{siq}\bZ_{iq}=\bZ_{siq}$, i.e. linkage errors have no impact on the sampled rows of $\bZ_{iq}$.
	We then have that
	\begin{equation}\label{jointexp}
	E_{\bA,M}(\by^\star_{siq}|\bX_{iq})=\bX^\star_{siq}\bbeta
	\end{equation}
	and 
	\begin{equation}\label{jointvar}
	V_{\bA,M}(\by^\star_{siq}| \bX_{iq})=\bSigma_{siq}=\bZ_{siq}\bSigma_{\bu_i} \bZ_{siq}^\prime+\bSigma_{seiq}+\bV_{siq},
	\end{equation}
	where $E_{\bA,M}$ and $V_{\bA,M}$ are the joint expectation and variance with respect to the linkage error model and the linear mixed model respectively. Here $\bSigma_{\bu_i}$ is a $m \times m$ matrix, corresponding to the $i$th diagonal block of $\bSigma_{u}$, and $\bV_{siq}=V_{\bA}(\bA_{siq}\bX_{iq}\bbeta)$. An exact expression for $\bV_{siq}$ is unavailable. However, using the arguments set out in \citet{OSR2009}, we can write down the approximation
	\begin{equation}\label{vsiq}
	\bV_{siq} \approx diag\{ v_{siq} \} = diag\left((1-\lambda_q)(\lambda_q (f_{iqj}-\bar{f}_{siq})^2+\bar{f}^{(2)}_{siq}  -\bar{f}_{siq}^2)\right), ~ j=1,\dots,n_{iq},
	\end{equation} 
	where  $\bf_{siq}=\{f_{iqj}\}=\bx_{ijq}^{\prime}\bbeta$ and $\bar{f}_{siq}$, $\bar{f}_{siq}^{(2)}$ denote the block $q$ averages of the components of $\bf_{siq}$ and their squares, respectively.

\begin{prop}\label{prop1}
		Under population model \eqref{EBLUP_model}, and assuming both non-informativeness of the sampling design, so the sample model \eqref{sample_mod} holds, and non-informativeness of the linkage error, as expressed in \eqref{noninfo}, it can be shown that:
		\begin{enumerate}
			\item[i)] the Best Linear Unbiased Estimator (BLUE) for $\bbeta$ given linked sample data is
			\begin{equation}\label{wls}
			\tilde{\bbeta}^\star =\left(\sum_{i=1}^{D} \sum_{q \in i}\bX_{siq}^{\star\prime} \bSigma_{siq}^{-1} \bX_{siq}^{\star}\right)^{-1}\left(\sum_{i=1}^{D}\sum_{q \in i}\bX_{siq}^{\star\prime} \bSigma_{siq}^{-1} \by_{siq}^{\star}\right),
			\end{equation}
			\item[ii)] the Best Linear Unbiased Predictor (BLUP) of $\bu_{i}$ given these data, and assuming an ELE linkage error structure within sub-population $U_{iq}$, is:
			\begin{eqnarray}\label{rablup}
			\tilde{\bu}_{i}^{\star}&=&\sum_{q \in i} \bSigma_{\bu_i} \bZ_{siq}^{\prime} \bSigma_{siq}^{-1} \left( \by^{\star}_{siq}-\bX^{\star}_{siq}\tilde\bbeta^{\star} \right)\label{rablup} \\
			&=&\sum_{q \in i} \bSigma_{\bu_i} \bZ_{siq}^{\prime} \bSigma_{siq}^{-1}\Big\{ \by^{\star}_{siq}-\bX_{siq}\tilde\bbeta^{\star}+ (1-\lambda_q)\frac{N_{iq}}{N_{iq}-1}\left(\bX_{iq}-\mathbf{1}_{n_{iq}}\bar{\bx}^\prime_{iq} \right)\tilde\bbeta^{\star}\Big\}.\label{eblup_res_1}
			\end{eqnarray} 
		\end{enumerate}
\end{prop}
	\begin{proof} 
		Formulas \eqref{wls} and \eqref{rablup} are modified version of the formulas obtained following standard BLUP derivations \citep[see][]{Rob91,
			RaoMol15}. See Appendix A for details. Note that we use the notation $q \in i$ above because we allow for the case where some areas only contain units from some of the blocks. Note also that \eqref{eblup_res_1} follows from \eqref{rablup} under the assumed ELE model since
		$
		\by^\star_{siq}-\bX^\star_{siq}\tilde{\bbeta}^\star 
		= \by^\star_{siq} - \big\{(\lambda_q-\gamma_q)\bX_{siq}+\gamma_qN_{iq}\uno_{n_{iq}}\bar{\bx}^\prime_{iq}   \big\}\tilde\bbeta^\star 
		$
		and $
		-\lambda_q\bX_{siq}\hat\bbeta^\star =(1-\lambda_q)\bX_{siq}\tilde\bbeta^\star - \bX_{siq}\tilde\bbeta^\star
		$.
		This implies
		$$
		\by^\star_{siq}-\bX^\star_{siq}\tilde{\bbeta}^\star 
		= \by^\star_{siq}-\bX_{siq}\tilde{\bbeta}^\star + (1-\lambda_q)\bX_{siq}\tilde{\bbeta}^\star
		+\gamma_q \bX_{siq}\tilde{\bbeta}^\star - \gamma_qN_{iq}\uno_{n_{iq}}\bar{\bx}^\prime_{iq}\tilde{\bbeta}^\star.
		$$
		As $\gamma_q=(1-\lambda_q)/(N_{iq}-1)$ and $(1-\lambda_q)=\gamma_q(N_{iq}-1)\gamma_q$ it follows that 
		$
		(1-\lambda_q)\bX_{siq}\tilde\bbeta^\star +\gamma_q\bX_{siq}\tilde\bbeta^\star =
		(1-\lambda_q)\frac{N_{iq}}{N_{iq}-1}\bX_{siq}\tilde\bbeta^\star
		$
		from which \eqref{eblup_res_1} easily follows.
	\end{proof}
	
	{\color{black} From (\ref{lem}) it is straightforward to see that the sum of the values making up $\by_{iq}$ will be the same as the corresponding sum for $\by_{iq}^\star$. Consequently the small area means $\bar{Y}_i$ and $\bar{Y}_i^{\star}$ will be identical, as will their respective BLUPs. Given (\ref{jointexp}) and (\ref{jointvar}), and the results set out in Proposition \ref{prop1} above, this BLUP (referred to from now on as the $^\star$BLUP) can be written
		\begin{equation}\label{BLUP_star}
		\tilde{\bar{Y}}_i^{\star BLUP}=N_i^{-1}\left\{ n_i \bar{y}_{si}^\star+(N_i-n_i)\left[ \bar{\bx}_{ri}^{\star\prime} \tilde{\bbeta}^{\star}+\bar\bz_{ri}^{\prime}\tilde{\bu}^{\star}  \right]\right\},
		\end{equation}
		where $\bar{y}_{si}^\star=n_{i}^{-1} \sum_{q \in i}\sum_{j \in s_{iq}}y_{iqj}^{\star}$ and $\bar{\bx}_{ri}^{\star}$ and $\bar\bz_{ri}$ denote the vectors of average values of $\bX_{i}^{\star}$ and $\bZ_{i}$ for the $N_{i}-n_{i}$ non-sampled units in the small area $i$.
	}

	The `empirical' versions of \eqref{wls} and \eqref{rablup}, denoted $\hat{\bbeta}^\star$ and $\hat{\bu}_{i}^{\star}$ respectively, are obtained by substituting estimates $\hat\bdelta$ for the unknown variance components $\bdelta$ that define the $\bSigma_{siq}$. {\color{black} These estimates can be obtained using either the pseudo maximum likelihood (pseudo-ML) or the pseudo restricted maximum likelihood (pseudo-REML) approach of \citet{Samartchambers2014}, and only depend on the conditional moments \eqref{jointexp} and \eqref{jointvar}. Note that the resulting estimators are called as pseudo-ML  (or pseudo-REML) because their estimating functions are based on the assumption that the matrices $\bSigma_{siq}$ are known.}
	The EBLUP of the the small area mean $\bar{Y}_i$ is then defined by substituting $\hat{\bbeta}^{\star}$ and $\hat{\bu}^{\star}$ for $\tilde{\bbeta}^{\star}$ and $\tilde{\bu}^{\star}$ respectively in (\ref{BLUP_star}). This is denoted $\hat{\bar{Y}}_i^{\star EBLUP}$, and referred to as the $^\star$EBLUP,  in what follows.
	
	Ignoring the $N_{iq}/(N_{iq}-1)$ term in \eqref{eblup_res_1}, we see that the $j$th component of the residual vector $\by^{\star}_{siq}-\bX^{\star}_{siq}\tilde\bbeta^{\star}$ is $y_{iqj}^{\star}-\bx_{iqj}^{\prime}\tilde\bbeta^{\star}+\left(1-\lambda_q)(\bx_{iqj}-\bar{\bx}_{iq}\right)^{\prime}\tilde{\bbeta}^{\star}$. That is, this modified residual can be interpreted as the naive residual that ignores linkage error plus a bias correction that increases in absolute value as the probability of an incorrect match increases and also as the leverage exerted by the individual fitted value $\bx_{iqj}^{\prime}\tilde{\bbeta}^{\star}$ increases.
	
	Alternatively, we can note the approximation
	\begin{equation}\label{eblup_res_2}
	\by^{\star}_{siq}-\bX^{\star}_{siq}\tilde\bbeta^{\star} \approx \lambda_q \left(\by^{\star}_{siq}-\bX_{siq}\tilde\bbeta^{\star}\right)+ (1-\lambda_q)\left(\by_{siq}^\star-\mathbf{1}_{n_{iq}}\bar{\bx}^\prime_{iq} \tilde\bbeta^{\star}\right).
	\end{equation}
	When the probability of an incorrect match is significantly greater than zero, the second term on the right hand side of \eqref{eblup_res_2} can be unstable. Consequently we consider an alternative expression for the predicted area effect that ignores this second term. This leads to a modified predictor of $\bu_i$ of the form
	
	\begin{equation}\label{rablup_mod}
	\tilde{\bu}_{i}^{\star\star}=\sum_{q \in i} \bSigma_{\bu_i} \bZ_{siq}^{\prime} \bSigma_{siq}^{-1} \lambda_q \left( \by^{\star}_{siq}-\bX_{siq}\tilde\bbeta^{\star} \right) .
	\end{equation} 
	
	The corresponding versions of the BLUP and EBLUP are referred to as $^{\star \star}$BLUP and $^{\star \star}$EBLUP respectively below.
	
	\subsection{MSE estimation for the $^\star$EBLUP}\label{sec:MSE:EBLUP}
	
	Methods for estimating the unconditional MSE of small area EBLUPs are typically based on averaging over the distribution of the random area effects, with the standard estimator of the unconditional MSE of the EBLUP predictor being the one suggested by \citet{Pra90}. In what follows we derive an estimator of the unconditional MSE of $\hat{\bar{Y}}_i^{\star EBLUP}$ along similar lines. In particular, we assume the regularity conditions (RCs) 1-6 set out in Appendix B and use the decomposition
	\begin{equation}\label{MSE_eblup}
	MSE_{\bA,M}(\hat{\bar{Y}}_i^{\star EBLUP})=MSE_{\bA,M}(\tilde{\bar{Y}}_i^{\star BLUP})+E_{\bA,M}\left((\hat{\bar{Y}}_i^{\star EBLUP}-\tilde{\bar{Y}}_i^{\star BLUP})^2\right).
	\end{equation}
	Under normality of the random area and individual effects, the cross-product term missing from \eqref{MSE_eblup} is zero provided the vector of variance components is translation invariant. Furthermore, the last component on the right hand side of \eqref{MSE_eblup} is generally intractable. We therefore approximate this term using a first-order Taylor expansion. Following \citet{Hen75}, we first note that
	the $MSE$ of $\tilde{\bar{Y}}_i^{\star BLUP}$ is 
	\begin{equation}\label{MSE_BLUP}
	MSE_{\bA,M}(\tilde{\bar{Y}}_i^{\star BLUP})=g_{1i}^*(\bdelta)+g_{2i}^*(\bdelta),
	\end{equation}
	where
	\begin{eqnarray}
	g_{1i}^*(\bdelta)=\bar\bz_i^{\prime}\Big[\bSigma_{\bu_i}-\bSigma_{\bu_i}\sum_{q \in i}\Big\{\bZ_{siq}^{\prime}\bSigma_{siq}^{-1} \bZ_{siq}\Big\}\bSigma_{\bu_i}\Big]\bar\bz_i;\\
	g_{2i}^*(\bdelta)=\bC_i \left(\sum_{i=1}^{D} \sum_{q \in i} \bX_{siq}^{\star\prime} {\bSigma}_{siq}^{-1} \bX_{siq}^{\star}\right)^{-1}  \bC_i^{\prime}.
	\end{eqnarray}
	Here $\bC_i=\bar{\bx}_i^{\prime}- \bar\bz_i^{\prime}\bSigma_{\bu_i}\sum_{q \in i}\Big\{\bZ_{siq}^{\prime}\bSigma_{siq}^{-1}\bX_{siq}^{\star}\Big\}$, $\bar\bx_i=N_{i}^{-1}\sum_{q \in i} N_{iq}\bar\bx_{iq}$ and $\bar\bz_{i}$ denotes the vector of average values of $\bZ_i$ for the $N_i$ units in the small area $i$. Next, a first order approximation to the second term on the right hand side of \eqref{MSE_eblup} follows from:
	
	\begin{prop}\label{prop2}
		We make the same assumptions as in Proposition \ref{prop1}. In addition we assume that regularity conditions 1-6 as specified in Appendix B hold, and that the random errors in \eqref{EBLUP_model} are normally distributed. Let $\hat{\bdelta}$ be the pseudo-REML estimator of $\bdelta$ and put
		\begin{equation}\label{g3}
		g_{3i}^*(\bdelta)=tr\Big[(\nabla \bb_i ^{\prime})\Big \{\sum_{q \in i}\bSigma_{siq}\Big\} (\nabla \bb_i ^{\prime})^{\prime}  E(\hat\bdelta-\bdelta)(\hat\bdelta-\bdelta)^{\prime} \Big],
		\end{equation}
		where $\nabla \bb_i ^{\prime}=col_{1\leqslant l \leqslant K}(\partial \bb_i ^{\prime}/\partial \delta_l)$ and $\bb_i ^{\prime}=\sum_{q \in i} \bSigma_{\bu_i} \bZ_{siq}^{\prime}\bSigma_{siq}^{-1}$. Then 
		\begin{equation}\label{g3f}
		E_{\bA, M}\left((\hat{\bar{Y}}_i^{\star EBLUP}-\tilde{\bar{Y}}_i^{\star BLUP})^2\right)=g_{3i}^*(\bdelta)+o(D^{-1}).
		\end{equation}
	\end{prop}
	\begin{proof}
		The expansion \eqref{g3f} is modified version of the one obtained by \cite{Pra90}.
	\end{proof}
	
	
	Following the same line of reasoning as in \citet{Pra90}, an approximately unbiased estimator of the MSE of $\hat{\bar{Y}}_i^{\star EBLUP}$ under \eqref{EBLUP_model} is then
	\begin{equation}\label{mse_eblup}
	\widehat{MSE}_{\bA,M}(\hat{\bar{Y}}_i^{\star EBLUP})=(1-\frac{n_i}{N_i})^2\{g_{1i}^*(\hat{\bdelta})+g_{2i}^*(\hat{\bdelta})+2g_{3i}^*(\hat{\bdelta})\},
	\end{equation}
	{\color{black} An estimator of the covariance matrix of the variance component estimator $\hat{\bdelta}$ is necessary in order to compute (\ref{g3}). This can be obtained as the inverse of the expected information matrix developed in \citet{Samartchambers2014}. If a pseudo-ML estimator is used instead, a bias correction to (\ref{mse_eblup}) is necessary, along the same lines as in \citet[][Section 6.2.6]{Rao03}.}
	Finally, we note that if the random area effects are estimated by \eqref{rablup_mod} then the three components of the estimated MSE become:
	\begin{equation}
	\label{g1mod} g_{1i}^{**}(\hat{\bdelta})= \bar\bz_i^{\prime}\Big[\hat\bSigma_{\bu_i}+\hat\bSigma_{\bu_i}\sum_{q \in i}\lambda_q^2\Big\{\bZ_{siq}^{\prime}\hat\bSigma_{siq}^{-1} \bZ_{siq}\Big\}\hat\bSigma_{\bu_i}-\hat\bSigma_{\bu_i}\sum_{q \in i}\lambda_q\Big\{\bZ_{siq}^{\prime}\hat\bSigma_{siq}^{-1} \bZ_{siq}\Big\}\hat\bSigma_{\bu_i}\Big]\bar\bz_i;
	\end{equation}
	\begin{equation}
	\label{g2mod} g_{2i}^{**}(\hat{\bdelta})=\hat\bC_i^{\star \star\prime} \left(\sum_{i=1}^{D} \sum_{q \in i} \bX_{siq}^{\star\prime} \hat{\bSigma}_{siq}^{-1} \bX_{siq}^{\star}\right)^{-1} \hat\bC_i^{\star \star};
	\end{equation}
	\begin{equation}
	\label{g3mod} g_{3i}^{**}(\hat{\bdelta})  = tr \left[ \nabla \bb_i ^{\star \star\prime}  \sum_{q \in i}\hat \bSigma_{siq} (\nabla \bb_i ^{\star \star\prime})^{\prime}  \hat{V}(\hat\bdelta)\right].
	\end{equation}
	Here $\hat\bC_i^{\star \star}=\bar{\bx}_i^{\prime}- \bar\bz_i^{\prime}\hat\bSigma_{\bu_i}\sum_{q \in i}\lambda_q \Big\{\bZ_{siq}^{\prime}\hat\bSigma_{siq}^{-1}\bX_{siq}^{\star}\Big\}$, $\nabla \bb_i ^{\star \star\prime}=col_{1\leqslant l \leqslant K}(\partial \bb_i ^{\star \star\prime}/\partial \delta_l)$ and\\ $\bb_i ^{\star \star\prime}=\sum_{q \in i}\lambda_q \bSigma_{\bu_i} \bZ_{siq}^{\prime}\bSigma_{siq}^{-1}$.
	
	%
	%
	\section{Robust linear mixed models for small area estimation with linked data}\label{sec:REBLUP}
	
	Outliers are a problem for any model-based survey estimation method, but particularly so for small area estimates. \citet{Sin09} proposed an estimator of a small area mean based on outlier robust estimation of the linear mixed model parameters. This robust-projective approach \citep{Cha14} uses plug-in robust prediction, i.e. the authors substitute outlier robust parameter estimates for the optimal, but outlier-sensitive, parameter estimates used in the EBLUP. In particular, they estimate fixed effects and variance components using a modified version of the estimating equations corresponding to the Robust ML Proposal II of \citet{Ric95} and compute outlier robust predictions for the random area effects using the robust estimating equations suggested by \citet{Fel86}. The solutions to these estimating equations depend on specification of a bounded influence function $\psi$, which we take to be the \citet{Hub81} influence function, defined as $\psi(u)=u\min(1,c/\vert c \vert)$ where $c$ is a tuning constant. Quantities that depend on this influence function (and hence on choice of the tuning constant) will be denoted by a superscript of $\psi$ below. The \citet{Sin09} robust version of the EBLUP, denoted REBLUP, of the small area mean $\bar{Y}_i$  is then:
	{\color{black} 
		\begin{equation}\label{REBLUP}
		\hat{\bar{Y}}_i^{\psi REBLUP}=N_i^{-1}\left\{n_i \bar{y}_{si}+(N_i-n_i)\left[ \bar{\bx}_{ri}^{\prime} \hat{\bbeta}^{\psi}+\bar\bz_{ri}^{\prime}\hat{\bu}^{\psi}  \right]\right\},
		\end{equation}
	}
	where $\hat{\bar{y}}_i= \sum_{j \in s_i}y_{ij}/n_i$, $\bar{\bx}_{ri}$ denotes the vector of average values of $\bX_{i}$ for the $N_{i}-n_{i}$ non-sampled units in the small area $i$, $\hat{\bbeta}^{\psi}$ is the robust estimated vector of regression coefficients and $\hat{\bu}^{\psi}$ is the the vector of robust predicted values of the area effects.

	Given linked data, we modify the estimating equations of both the Robust ML Proposal II of \citet{Ric95} and of  \citet{Fel86} to account for linkage errors, making the same assumptions (ELE errors, one to one, complete and non-informative linkage) as in Section \ref{sec:ass}. The modified \citet{Fel86} estimating equations assume the variance components $\bdelta$ are known, and define robust estimates $\bbeta^{\star}$ and $\bu^{\star}$ of the fixed effects and the area random effects respectively. They are:
	\begin{equation}
	\label{rob1_star} \sum_{i=1}^D \sum_{q \in i}\bX_{siq}^{\star \prime}\bSigma_{siq}^{-1}\bU_{siq}^{1/2}\psi\{\br_{siq}^{\star}\}= \bzero,
	\end{equation}
	and
	\begin{equation}
	\label{rra} \sum_{i=1}^D \sum_{q \in i} \left \{\bZ_{siq}^\prime\bSigma_{seAiq}^{-1/2} \psi\left\{\bSigma_{seAiq}^{-1/2} (\by_{siq}^{\star}-\bX_{siq}^{\star}  \bbeta^{\star}-\bZ_{siq} \bu_{i}) \right\}- \bSigma_{\bu_i}^{ -1/2} \psi\{ \bSigma_{\bu_i}^{-1/2}\bu_i  \} \right\} =\bzero.
	\end{equation}
	Here $\br_{siq}^{\star}=\bU_{siq}^{-1/2}(\by_{siq}^{\star}-\bX_{siq}^{\star}\bbeta^{\star})$, $\bSigma_{seAiq}=\bSigma_{seiq}+\bV_{siq}$, $\bSigma_{siq}=\bZ_{siq}\bSigma_{\bu_i}\bZ_{siq}^\prime+\bSigma_{seAiq}$ where $\bV_{siq}$ is defined in \eqref{vsiq} and $\bU_{siq}$ is a diagonal matrix with the same diagonal entries as $\bSigma_{siq}$. Robust estimates of the variance components $\bdelta$ are obtained as solutions to the modified version of the Robust ML Proposal II estimating equations of \citet{Ric95}, and are then substituted in (\ref{rob1_star}) and (\ref{rra}). Put $\bK_{siq}=E\{\psi^2(R )\}\bI_{n_{iq}}$, where $R$ is a standard normal random variable. These estimating equations are then
	\begin{equation}
	\label{rob2_star} \sum_{i=1}^D\sum_{q \in i} \left\{\ppsi^{\prime}\{\br_{siq}^{\star}\}\bU_{siq}^{1/2} \bSigma_{siq}^{-1} \frac{\partial \bSigma_{siq}}{\partial \delta_l} \bSigma_{siq}^{\-1}  \bU_{siq}^{1/2} \ppsi\{\br_{siq}^{\star}\} -  tr\left( \bK_{siq} \bSigma_{siq}^{-1}\frac{\partial \bSigma_{siq}}{\partial \delta_l}  \right) \right\}=\bzero,
	\end{equation}
	and the corresponding linkage error-adjusted version of the REBLUP of \citet{Sin09} is 
	\begin{equation}\label{REBLUP_star}
	\hat{\bar{Y}}_i^{\psi \star REBLUP}=N_i^{-1}\left\{n_i\bar{y}_{si}^{\star}+(N_i-n_i)\left[ \bar{\bx}_{ri}^{\star\prime} \hat{\bbeta}^{\psi \star}+\bar\bz_{ri}^{\prime}\hat{\bu}^{\psi \star}  \right]\right\}.
	\end{equation}
	Here $\hat{\bbeta}^{\psi \star}$ and $\hat{\bu}^{\psi \star}$ depend on the influence function $\psi$, and are the solutions to the robust estimating equations (\ref{rob1_star}) and (\ref{rra}) respectively, using substituted values of $\bdelta$ obtained by solving (\ref{rob2_star}). Note that the value of (\ref{REBLUP_star}) when the true value of $\bdelta$ is used in (\ref{rob1_star}) and (\ref{rra}) is referred to below as the RBLUP, denoted $\tilde{\bar{Y}}_i^{\psi \star RBLUP}$, with the corresponding solutions to (\ref{rob1_star}) and (\ref{rra}) then denoted by $\tilde{\bbeta}^{\psi \star}$ and $\tilde{\bu}^{\psi \star}$ respectively. We refer to these linkage-error adjusted versions of RBLUP and REBLUP as $^\star$RBLUP and $^\star$REBLUP respectively below.

	\subsection{MSE estimation for the $^\star$REBLUP}\label{mse_REBLUP}
	We develop an analytic estimator for the MSE of (\ref{REBLUP_star}) under a working mixed model that conditions on the realized values of the area effects, i.e. the proposed MSE estimator is an estimator of the conditional MSE of $\hat{\bar{Y}}_i^{\psi \star REBLUP}$. It is based on the assumption that a consistent estimator of the MSE of a linear approximation to a non-linear small area estimator can be used as an estimator of the MSE of that small area estimator. See \citet{booth98} and \citet{Cha14}. Such linearization-based MSE estimators are generally not consistent, and can be biased low, see \citet{Har92}. However, in small sample problems this is not an issue since it is the variability, rather than the bias, of the MSE estimator that is of concern. The development below omits some technical details, which are available from the authors on request. 
	
	\begin{prop}
		Put $\tilde{\btheta}^{\psi \star}=(\tilde{\bbeta}^{\psi \star \prime},\tilde{\bu}^{\psi \star \prime})^\prime$ and assume that this random variable converges in probability to $\btheta_0^{\psi \star}=({\bbeta}_0^{\psi \star \prime}, {\bu}_0^{\psi \star \prime})^\prime$. Also let $V_{\bA, M | \bu}$ denote variance with respect to both the linkage error model and the linear mixed model for $\by$ in terms of $\bX$ given the realized values of the area effects. Suppose the same assumptions are made as in Proposition \ref{prop1}, and in addition, regularity conditions 1-5 and 7-9 in Appendix B apply and $\psi$ corresponds to the Huber influence function. Then the conditional prediction variance of the {\color{black}$^\star$RBLUP} can be expressed as
		{\color{black} 
			\begin{equation}\label{MSE_RBLUP}
			V_{\bA, M | \bu}(\tilde{\bar{Y}}_i^{\psi \star RBLUP}-\bar{Y}_i)=\left(1-\frac{n_i}{N_i}\right)^2 \left\{\begin{array}{c|c}(\bar{\bx}_{ri}^{\star\prime}~~ & ~~\bar{\bz}_{ri}^\prime)\end{array}V_{\bA ,M| \bu}(\tilde{\btheta}^{\psi \star})\begin{array}{c|c}(\bar{\bx}_{ri}^{\star\prime}~~ & ~~\bar{\bz}_{ri}^\prime)\end{array}^{\prime}+V_{\bA, M| \bu}(\bar\be^{\psi \star}_{ri})\right\}+o(D^{-1})
			\end{equation}
			where 
			$\bar\be^{\psi \star}_{ri}=(N_i-n_i)^{-1}\sum_{j \in r_i}(y_{ij}^{\star}-\bx_{ij}^{\star\prime}{\bbeta}_0^{\psi \star }-\bz_{ij}\bu_0^{\psi \star})$.}
	\end{prop}
	\begin{proof}
		Formula \eqref{MSE_RBLUP} is a modified version of a similar prediction variance formula developed in \cite{Cha14}.
	\end{proof}
	
	
	A first-order approximation to the prediction variance (\ref{MSE_RBLUP}) is then:
	\begin{equation}\label{var_u_star}
	\hat{V}_{\bA, M|\bu}(\tilde{\bar{Y}}_i^{\psi \star RBLUP}-\bar{Y}_i)=h_{1i}(\tilde{\btheta}^{\psi \star})+h_{2i}(\tilde{\btheta}^{\psi \star}),
	\end{equation}
	where
	\begin{enumerate}
		{\color{black} 
			\item $h_{1i}(\tilde{\btheta}^{\psi \star})=\left(1-\frac{n_i}{N_i}\right)^2 \begin{array}{c|c}(\bar{\bx}_{ri}^{\star\prime}~~ & ~~\bar{\bz}_{ri}^\prime)\end{array}\hat{V}_{\bA, M| \bu}(\tilde{\btheta}^{\psi \star})\begin{array}{c|c}(\bar{\bx}_{ri}^{\star\prime}~~ &~~ \bar{\bz}_{ri}^\prime)\end{array}^{\prime}$. Here $\hat{V}_{\bA, M| \bu}(\tilde{\btheta}^{\psi \star})$ is the sandwich-type estimator of $V_{\bA M| \bu}(\tilde{\btheta}^{\psi \star})$ set out in Appendix C.
			\item $h_{2i}(\tilde{\btheta}^{\psi \star})= \left(1-\frac{n_i}{N_i}\right)^2 \hat{V}_{\bA, M| \bu}(\bar\be^{\psi \star}_{ri})$ where
			$$\hat{V}_{\bA, M|\bu}(\bar\be^{\psi \star}_{ri})=(N_i-n_i)^{-1}(n_i-1)^{-1}\sum_{l}\sum_{j \in s_l} (y_{lj}^{*}-\bx_{lj}^{\star\prime}\tilde\bbeta^{\psi \star}-\bz_{lj}^\prime \tilde\bu^{\psi \star})^2.$$}
		Note that $\hat{V}_{\bA, M|\bu}(\bar\be^{\psi \star}_{ri})$ above pools data from the entire sample. This leads to more stable MSE estimates when area sample sizes are very small.
	\end{enumerate}
	The estimator of the MSE of the {\color{black}$^\star$RBLUP} is obtained by adding an estimator of the squared conditional bias to \eqref{var_u_star}: 
	\begin{equation}\label{mse_u_star}
	\widehat{MSE}_{\bA ,M |\bu}(\tilde{\bar{Y}}_i^{\psi \star RBLUP})=h_{1i}(\tilde{\btheta}^{\psi \star})+h_{2i}(\tilde{\btheta}^{\psi \star})+\hat B_{\bA, M| \bu}^{2}(\tilde{\bar{Y}}_i^{\psi \star RBLUP}),
	\end{equation}
	where
	{\color{black}
		\begin{equation}\label{bias_u_star}
		\hat B_{\bA, M| \bu}(\tilde{\bar{Y}}_i^{\psi \star RBLUP})=\sum_{j \in s} w_{ij}^{\psi \star RBLUP}\tilde\mu_{ij}^{\star}-N_{i}^{-1}\sum_{j \in U_i}\tilde\mu_{ij}^{\star}.
		\end{equation}
		Here $\tilde\mu_{ij}^{\star}$ is an unbiased linear estimator of the conditional expected value $\mu_{ij}^{\star}=E_{\bA, M}(y_{ij}^{\star}|\bx_{ij}, \bu^{\psi\star})$ and $w_{ij}^{\psi \star RBLUP}$} is the weight for unit $j$ in area $i$ based on the pseudolinearisation approach to MSE estimation that was described by \citet{Cha11a} and extended to RBLUP by \citet{Cha14}. These weights $w_{ij}^{\psi \star RBLUP}$ can be obtained in a straightforward manner in the case of linkage data following \citet{Cha11a}.
	
	Finally, we consider MSE estimation for the {\color{black}$^\star$REBLUP} \eqref{REBLUP_star}, noting that an estimator of its conditional MSE can based on a decomposition similar to that used in \citet{Pra90}:
	\begin{equation}\label{MSE_rob}
	MSE_{\bA,M|\bu}(\hat{\bar{Y}}_i^{\psi \star REBLUP})=MSE_{\bA,M|\bu}(\tilde{\bar{Y}}_i^{\psi \star RBLUP})+E_{\bA,M|\bu}\left((\hat{\bar{Y}}_i^{\psi \star REBLUP}-\tilde{\bar{Y}}_i^{\psi \star RBLUP})^2\right)+O(D^{-1}).
	\end{equation}
	An approximation to the second term on the right-hand side of equation \eqref{MSE_rob} can be obtained as follows:
	
	\begin{prop}
		We make the same assumptions as in Proposition \ref{prop1}. In addition, we assume that the regularity conditions 1-5 and 7-9 in Appendix B hold. Let $\hat{\bdelta}^{\psi \star}$ be the vector of estimated variance components obtained by solving the robust estimating equations (\ref{rob1_star}) - (\ref{rob2_star}). Then 
		\begin{equation}\label{h3}
		E_{\bA,M|\bu}\left((\hat{\bar{Y}}_i^{\psi \star REBLUP}-\tilde{\bar{Y}}_i^{\psi \star RBLUP})^2\right)=h_{3i}(\hat{\bdelta}^{\psi \star})+O(D^{-1}),
		\end{equation}
		where
		\begin{equation}
		\label{h3i}
		h_{3i}(\hat{\bdelta}^{\psi \star})=\left(N_i^{-1}\sum_{j \in r_i}\bz_{ij}^{\prime}  \right)\sum_{q \in i}\bOmega_{siq} V_{\bA, M| \bu}(\hat{\bdelta}^{\psi \star})\left(N_i^{-1}\sum_{j \in r_i}\bz_{ij}^{\prime}  \right)^{\prime}.
		\end{equation}
		Here
		$$\bOmega_{siq}=\sum_{k=1}^{K}\sum_{g=1}^{K}\left\{ (\partial_{\delta_k} \bB_{siq}) \left[ \sum_j \sum_l \{(\bz_{ij}^{\prime}\bu_0^{\psi \star})(\bz_{il}^{\prime}\bu_0^{\psi \star}) \} +\bSigma_{seAiq}\right]  (\partial_{\delta_g} \bB_{siq})^{\prime}\right\},$$ 
		
		\begin{eqnarray}
		\bB_{siq} & = & \left(\bZ_{siq}^{\prime}\bSigma_{seAiq}^{-1/2}\bW_{2siq}\bSigma_{seAiq}^{-1/2}\bZ_{siq}+\bSigma_{\bu_i}^{-1/2}\bW_{3siq}\bSigma_{\bu_i}^{-1/2}\right)^{-1} \\
		\label{BSIQ}& & \left(\bZ_{siq}^{\prime}\bSigma_{seAiq}^{-1/2}\bW_{2siq}\bSigma_{seAiq}^{-1/2}\right),
		\end{eqnarray}
		
		$$\bW_{2siq}= \psi\left\{\bSigma_{seAiq}^{-1/2} (\by_{siq}^{\star}-\bX_{siq}^{\star} \tilde \bbeta^{\psi \star } -\bZ_{siq} \tilde \bu_{i}^{\psi \star}) \right\} \left\{\bSigma_{seAiq}^{-1/2} (\by_{siq}^{\star}-\bX_{siq}^{\star}  \tilde\bbeta^{\psi \star} -\bZ_{siq} \tilde\bu_{i}^{\psi \star}) \right\}^{-1}$$ is a $n_{iq} \times n_{iq}$ diagonal matrix of weights for the individual effects in area $i$, and 
		$$\bW_{3siq}= \psi\{ \bSigma_{\bu_i}^{-1/2}\tilde\bu_i^{\psi \star}  \}  \left\{\bSigma_{\bu_i}^{-1/2}\tilde\bu_i^{\psi \star}\right\}^{-1}   $$
		is a $m \times m$ diagonal matrix of weights for the area effect associated with area $i$.
	\end{prop}
	\begin{proof}
		Formula \eqref{g3} takes into account expectation with respect to the linkage error model and the model that relates $\by$ and $\bX$, and is a suitably modified version of a similar formula in \cite{Cha14}.
	\end{proof}
	We define an estimator $\hat{V}_{\bA, M |\bu}(\hat\bdelta^{\psi \star})$ of the variance-covariance matrix of the estimated variance components in Appendix C, based on the approach taken by \citet{Sin09}. Let $\hat{h}_{1i}(\hat{\btheta}^{\psi \star})$ and $\hat{h}_{2i}(\hat{\btheta}^{\psi \star})$ denote the values of $h_{1i}(\tilde{\btheta}^{\psi \star})$ and $h_{2i}(\tilde{\btheta}^{\psi \star})$ respectively when all unknown parameters are replaced by robust estimates, and put $\hat{h}_{3i}(\hat{\bdelta}^{\psi \star})$ equal to the value of (\ref{h3i}) when the estimator $\hat{V}_{\bA, M |\bu}(\hat\bdelta^{\psi \star})$ is substituted. An estimator of the conditional MSE of the {\color{black}$^\star$REBLUP} for area $i$ is then:
	\begin{equation}\label{est_mse_u_star}
	\widehat{MSE}_{\bA ,M| \bu}(\hat{\bar{Y}}_i^{\psi \star REBLUP})=\hat{h}_{1i}(\hat{\btheta}^{\psi \star})+\hat{h}_{2i}(\hat{\btheta}^{\psi \star})+\hat{h}_{3i}(\hat{\btheta}^{\psi \star})+\hat B_{\bA, M| \bu}^{2}(\hat{\bar{Y}}_i^{\psi \star REBLUP}).
	\end{equation}

	%
	%

	\section{M-quantile models for small area estimation with linked data}\label{sec:MQ}
	
	M-quantile regression models were first suggested for small area estimation by \cite{Cha06}. See \cite{Bianchi18} for a recent review of subsequent applications and theoretical extensions. Here we briefly discuss basic ideas and develop appropriate notation. 
	
	
	
	\citet{Bre88} introduced M-quantile regression as a `quantile-like' generalization of regression based on a bounded influence function $\psi$ \citep{Hub81}, with associated loss function $\rho$ such that $\psi=d\rho(u)/du$. For a given $\tau \in (0,1)$, the  M-quantile of order $\tau$ of a random variable is defined as the value minimizing the expectation of the tilted loss function $\rho_\tau = \vert \tau - I(u<0) \vert \rho (u)$, where $\rho(u)$, $u \in \Re$, is continuously differentiable with $\rho(0)=0$. M-quantiles aim at combining the robustness properties of quantiles ($\rho(u)=|u|$) with the efficiency properties of expectiles ($\rho(u)=u^2$). By definition, any M-quantile must depend on specification of its influence function $\psi$, and so we will not explicitly refer to $\psi$ in our notation below for quantities that depend on the values of M-quantiles that are all defined using the same $\psi$.
	
	In the linear case, M-quantile regression leads to a family of hyperplanes indexed by a real number $\tau \in (0,1)$ representing the order of the M-quantile of interest, i.e.
	\begin{displaymath}
	MQ(\tau|\bx_{ij})=\bx_{ij}^{\prime}\bbeta_\tau.
	\end{displaymath}
	For specified $\tau$ and influence function $\psi$ (with $\psi_\tau=d\rho_\tau(u)/du$) an estimate $\hat{\bbeta}_{\tau }$ of the vector of regression parameters $\bbeta_\tau$ may be obtained as the solution to the normal equations,
	\begin{equation}
	\label{mqnormal}
	\sum_{i=1}^{D}\sum_{j=1}^{n_i}\psi_\tau \left(\frac{y_{ij}-\bx_{ij}^{\prime}\hat{\bbeta}_{\tau }}{\sigma_\tau}\right)\bx_{ij}=\bzero,
	\end{equation}
	where $\sigma_\tau$ is a scale parameter that characterizes the spread of the distribution of the residuals $y_{ij}-\bx_{ij}^{\prime}\bbeta_\tau$. Following standard practice in robust M-regression, this scale parameter can be estimated by $\hat{\sigma}_{\tau}=median\{|y_{ij}-\bx_{ij}^{\prime}\hat{\bbeta}_{\tau}|\}/0.6745$. When $\psi$ is a continuous function, \citet{Bre88} adapt the iteratively reweighted least squares (IRWLS) approach to linear M-regression and show that a linear M-quantile regression of specified order $\tau$ can be fitted by weighting positive residuals from the M-quantile line by $\tau$ and negative residuals by $1-\tau$. Note that in what follows we will assume that $\psi$ is the Huber influence function, with tuning constant $c>0$, and so $\psi$ is continuous.
	

	Following \cite{Cha06}, we characterize conditional variability given $\bx_{ij}$ across the population of interest by the M-quantile coefficients of the population units. For unit $j$ in area $i$ this coefficient is the value $\tau_{ij}$ such that $MQ(\tau_{ij}|\bx_{ij})=y_{ij}$. If a hierarchical structure does explain part of the variability in the population data, units within areas defined by this hierarchy are expected to have similar M-quantile coefficients. When the conditional M-quantiles are assumed to follow a linear model, with $\bbeta_\tau$ a sufficiently smooth function of $\tau$, \citet{Cha06} suggest a predictor of $\bar{Y}_i$ of the form
	{\color{black} 
		\begin{equation}\label{MQ_predictor}
		\hat{\bar{Y}}_i^{MQ}=N_{i}^{-1}\left\{n_i\bar{y}_{si}+(N_i-n_i)\bar{\bx}_{ri}^{\prime} \hat{\bbeta}_{\hat{\tau}_i}\right\},
		\end{equation}}
	where $\hat{\tau}_i=n_i^{-1}\sum_{j \in s_i}\hat \tau_{ij}$ is an estimate of the average value of the M-quantile coefficients $\tau_{ij}$ for population units in area $i$, and $\hat \tau_{ij}$ is defined by the estimating equation $y_{ij}=\bx_{ij}^\prime \hat \bbeta_{\hat{\tau}_{ij}}$.

	Naive use of M-quantile regression modelling when the data contain linkage errors leads to biased estimates of the true M-quantile fits. This is intuitively clear from the fact that the combined impact of natural variability as well as linkage error variability leads to conditional distributions at the different $\bx_{ij}$ that are not the ones of interest. It is also clear from a cursory inspection of (\ref{jointexp}) and (\ref{jointvar}). Following the approach of \citet{OSR2009}, we therefore modify the M-quantile normal equations (\ref{mqnormal}) to take account of the linkage error structure, using the notation introduced in Section \ref{sec:ass}. This leads to the modified M-quantile normal equations,

	\begin{equation}\label{est_mq}
	\sum_{i=1}^D\sum_{q \in i} \bX_{siq}^{\star\prime}\bUpsilon_{siq\tau}^{-1/2}\psi_\tau \left \{ \bUpsilon_{siq\tau}^{-1/2}(\by_{siq}^{\star}-\bX_{siq}^{\star\prime}\bbeta_{\tau}^{\star})  \right\}=\bzero,
	\end{equation}
	where $\bUpsilon_{siq\tau}=diag\left(\sigma_{\tau}^{\star 2}+(1-\lambda_q)(\lambda_q (f_{iqj\tau}-\bar{f}_{siq\tau})^2+\bar{f}^{(2)}_{siq\tau}  -\bar{f}_{siq\tau}^2)\right)$, $\bf_{siq\tau}=\{f_{iqj\tau}\}=\bx_{iqj}^{\prime}\bbeta_{\tau}^{\star}$, and $\bar{f}_{siq\tau}$, $\bar{f}_{siq\tau}^2$ denote the block $q$ averages of the components of $\bf_{siq\tau}$ and their squares respectively. Note that $\sigma_{\tau}^{\star}$ here is the scale coefficient associated with the skewed residuals from the M-quantile regression line of order $\tau$. Given $\bUpsilon_{siq\tau}$, the solution to \eqref{est_mq} can be obtained via IRWLS, and is of the form
	\begin{equation}\label{beta_star_mq}
	\tilde{\bbeta}_\tau^{\star}=\left(\sum_{i=1}^{D}\sum_{q \in i}\bX_{siq}^{\star\prime}\bUpsilon_{siq\tau}^{-1/2}\bW_{siq\tau}^{\star} \bUpsilon_{siq\tau}^{-1/2}\bX_{siq}^{\star}\right)^{-1}\left(\sum_{i=1}^D\sum_{q \in i}\bX_{siq}^{\star \prime}\bUpsilon_{siq\tau}^{-1/2}\bW_{s iq\tau}^{\star} \bUpsilon_{siq\tau}^{-1/2} \by_{siq}^{\star}\right),
	\end{equation}
	where $\bW_{siq\tau}^{\star}$ is a diagonal matrix of weights defined by component-wise division of the vector $\psi_\tau \left\{\bUpsilon_{siq\tau}^{-1/2} ( \by_{siq}^{\star}-\bX_{siq}^{\star}\tilde{\bbeta}_{\tau}^{\star}) \right\}$ by the vector $\bUpsilon_{siq\tau}^{-1/2} ( \by_{siq}^{\star}-\bX_{siq}^{\star}\tilde{\bbeta}_{\tau}^{\star})$. Similarly, given $\bbeta_{\tau}^{\star}$ and $\bUpsilon_{siq\tau}$, a robust estimator of $\sigma_{\tau}^{\star 2}$ is
	{\color{black}
		\begin{equation}\label{sigma_star_mq}
		\tilde{\sigma}_{\tau}^{\star 2}=\left(\sum_{i=1}^D\sum_{q \in i}tr(\bW_{si q\tau}^{\star})  \right)^{-1}\sum_{i=1}^{D}\sum_{q \in i}\left\{ (\by_{siq}^{\star}-\bx_{iqj}^{\star \prime}\tilde{\bbeta}_{\tau}^{\star})^\prime  \bW_{s iq\tau}^{\star} (\by_{siq}^{\star}-\bx_{iqj}^{\star \prime}\tilde{\bbeta}_{\tau}^{\star})\right\}.
		\end{equation}
	}
	The `empirical' versions of $\tilde{\bbeta}_\tau^{\star}$ and $\tilde{\sigma}_{\tau}^{\star 2}$, which we denote by $\hat{\bbeta}_\tau^{\star}$ and $\hat{\sigma}_{\tau}^{\star 2}$ respectively, are then defined by iterating between (\ref{beta_star_mq}) and (\ref{sigma_star_mq}).
	
	In order to use the M-quantile approach for small area estimation, it is first necessary to estimate the M-quantile coefficients defined by the correctly linked sample values $\by_{siq}$, i.e. the values $\hat{\tau}_{iqj}$ such that $y_{iqj}=\bx_{iqj}^\prime \hat{\bbeta}_{\hat{\tau}_{iqj}}$ for $j \in s_{iq}$.  Unfortunately, replacing $y_{iqj}$ by its linked value $y_{iqj}^\star$ leads to biased estimates of these coefficients. We therefore propose to use an approximation to $\hat{\tau}_{iqj}$ based on linked data that is corrected for linkage error induced bias. Let $\hat{\tau}_{iqj}^{\star \star}$ satisfy $y_{iqj}^\star=\bx_{iqj}^\prime \hat{\bbeta}_{\hat{\tau}_{iqj}^{\star \star}}$. We then define our linked data-based estimate of the M-quantile coefficient for $j \in s_{iq}$ as $\hat\tau_{iqj}^{\star}=(\lambda_q-\gamma_q)\hat \tau_{iqj}^{\star\star}+\gamma_{q} N_{iq} \hat{\bar{\tau}}_{iq}^{\star\star}$ where $\hat{\bar{\tau}}_{iq}^{\star\star}=n_{iq}^{-1}\sum_{k \in s_{iq}}\hat\tau_{iqk}^{\star\star}$. If the values $\hat\tau_{iqk}^{\star\star}$ are spread over $(0,1)$ the M-quantile coefficient for unit $j \in s_{iq}$ can be approximated by $\hat\tau_{iqj}^{\star}\approx (\lambda_q-\gamma_q)\hat\tau_{iqj}^{\star\star}+\gamma_{q} N_{iq} 0.5$. An estimated area-specific M-quantile coefficient is then computed as $\hat{\tau}_i^\star= n_i^{-1}\sum_{j \in s_{iq}} \hat{\tau}_{iqj}^\star$, and the M-quantile predictor of the area $i$ mean $\bar{Y}_i$ using linked data (the $^\star$M-quantile predictor) becomes
	{\color{black}
		\begin{equation}\label{MQ_predictor_linked}
		\hat{\bar{Y}}_i^{\star MQ}=N_i^{-1}\left\{ n_i \bar{y}_{si}^{\star}+(N_i-n_i)\bar{\bx}_{ri}^{\star\prime}\hat{\bbeta}_{\hat{\tau}_i^{\star}}^{\star}\right\}
		\end{equation}
	}

	\subsection{MSE estimation for the $^\star$M-quantile predictor}\label{sec:mse:mq}
	In this section we develop an MSE estimator for $\hat{\bar{Y}}_i^{\star MQ}$ based on the linearisation approach of \citet{Cha14}. Conditioning on the value of $\hat{\tau}_i^{\star}$, the prediction variance of \eqref{MQ_predictor_linked} is
	{\color{black}
		\begin{equation}\label{var_pred_MQ}
		V_{\bA,M}(\hat{\bar{Y}}_i^{\star MQ}-\bar{Y}_i|\hat{\tau}_i^{\star})=(1-n_i/N_i)^2\left( \bar{\bx}_{ri}^{\star\prime} V_{0 \bA,M}(\hat{\bbeta}_{\hat{\tau}_i^{\star}}^{\star})\bar{\bx}_{ri}^{\star} +V_{0 \bA,M}(\bar{e}_{ri}^{\star})\right),
		\end{equation}
	}
	where the subscript of $0$ defines true values under this area-specific model. Put $E_{0 \bA,M}(\hat{\tau}_i^{\star})=\tau_{0i}^{\star}$. A first-order approximation to $V_{0 \bA, M}(\hat{\bbeta}_{\hat{\tau}_i^{\star}}^{\star})$ is then
	\begin{equation}\label{var_beta0}
	V_{0 \bA, M}(\hat{\bbeta}_{\hat{\tau}_i^{\star}}^{\star})=[E_{0 \bA,M}(\partial_{\bbeta_\tau^\star} \bH|\tau=\tau_{0i}^{\star})]^{-1}V_{0 \bA,M}(\bH(\bbeta_{\tau_{0i}^{\star}}^{\star}))\left( [E_{0 \bA,M}(\partial_{\bbeta_\tau^\star} \bH|\tau=\tau_{0i}^{\star})]^{-1} \right)^{\prime}+o(n^{-1}).
	\end{equation}
	An estimator of \eqref{var_beta0} is
	\begin{eqnarray}
	\nonumber \hat{V}_{\bA,M}(\hat{\bbeta}_{\hat{\tau}_i^{\star}}^{\star})&=&\frac{n}{n-p}\left(\sum_{i=1}^{D}\sum_{q \in i}\bX_{siq}^{\star\prime}\hat\bUpsilon_{siq\tau}^{-1/2}\hat \bPsi_{siq \tau} \hat\bUpsilon_{siq\tau}^{-1/2}\bX_{siq}^{\star}\right)^{-1} \left(\sum_{i=1}^{D}\sum_{q \in i}\bX_{siq}^{\star\prime}\hat\bUpsilon_{siq\tau}^{-1/2}   \hat\bPhi_{siq \tau}   \hat\bUpsilon_{siq\tau}^{-1/2}\bX_{siq}^{\star}\right)\\
	& &\left( \left(\sum_{i=1}^{D}\sum_{q \in i}\bX_{siq}^{\star\prime}\hat\bUpsilon_{siq\tau}^{-1/2}\hat \bPsi_{siq \tau}\hat\bUpsilon_{siq\tau}^{-1/2}\bX_{siq}^{\star}\right)^{-1}\right)^\prime,
	\end{eqnarray}
	where $\hat\bUpsilon_{siq\tau}$ is the plug-in estimate of $\bUpsilon_{siq\tau}$ defined by $\tau=\hat\tau_i^\star$, with $\hat \bPsi_{siq \tau}=diag\left(\psi_{\tau}^{\prime} \left \{ \hat \bUpsilon_{siq\tau}^{-1/2}(\by_{siq}^{\star}-\bX_{siq}^{\star\prime}\hat{\bbeta}_{\hat{\tau}_i^{\star}}^{\star})  \right\}\right) $ and $\hat\bPhi_{siq \tau}=diag\left(\psi_{\tau}^2 \left \{\hat \bUpsilon_{siq\tau}^{-1/2}(\by_{siq}^{\star}-\bX_{siq}^{\star\prime}\hat{\bbeta}_{\hat{\tau}_i^{\star}}^{\star})  \right\}\right)$. An estimator of \eqref{var_pred_MQ} is then
	\begin{equation}\label{est_var_pred_MQ}
	\hat V_{\bA,M}(\hat{\bar{Y}}_i^{\star MQ})=(1-n_i/N_i)^2\left( \bar{\bx}_{ri}^{\star \prime} \hat{V}_{\bA,M}(\hat{\bbeta}_{\hat{\tau}_i^{\star}}^{\star})\bar{\bx}_{ri}^\star +\hat{V}_{\bA,M}(\bar{e}_{ri}^{\star})\right),
	\end{equation}
	where 
	{\color{black}
		\begin{equation}
		\hat{V}_{\bA,M}(\bar{e}_{ri}^{\star})=(N_i-n_i)^{-1}(n-1)^{-1}\sum_{h=1}^D\sum_{j \in s_h}(y_{hj}^{\star}-\bx_{hj}^{\star\prime}\hat{\bbeta}_{\hat{\tau}_h^{\star}}^{\star} )^2.
		\end{equation}
	}
	An estimator of the area-specific bias of $\hat{\bar{Y}}_i^{\star MQ}$ is
	{\color{black}
		\begin{equation}\label{biasMQ}
		\hat{B}_{\bA, M}(\hat{\bar{Y}}_i^{\star MQ})=N_{i}^{-1}\left( \sum_{h=1}^D\sum_{q \in h}\sum_{j \in s_{hq}}w_{iqj}\bx_{hqj}^{\star\prime}\hat{\bbeta}_{{\hat{\tau}}_h^{\star}}^{\star}-\sum_{j \in U_i}\bx_{ij}^{\star\prime}\hat{\bbeta}_{{\hat{\tau}}_i^{\star}}^{\star}   \right),
		\end{equation}
	}
	{\color{black}
		where $w_{iqj}=b_{iqj}+I(j \in i \bigcap q)$ and 
		\begin{displaymath}
		\bb_{iq}=\{b_{iqj}\}=\hat\bUpsilon_{siq\tau}^{-1/2}\hat\bW_{siq\tau}^{\star}\hat\bUpsilon_{siq\tau}^{-1/2}\bX_{siq}^{*}\left(\bX_{siq}^{*\prime}\hat \bUpsilon_{siq\tau}^{-1/2} \hat\bW_{siq \tau}^{\star}\hat\bUpsilon_{siq\tau}^{-1/2}\bX_{siq}^{*}\right)^{-1}(N_{iq}-n_{iq})\bar{\bx}_{riq}^{\star},
		\end{displaymath}
		where $\hat\bW_{siq\tau}^\star$ is the plug-in estimate of $\bW_{siq\tau}^\star$ defined by $\tau=\hat\tau_i^\star$, and where $\bar{\bx}_{riq}^{\star}$ denotes the vector of column averages of $\bX_{iq}^{\star}$ restricted to the $N_{iq}-n_{iq}$ non-sampled units in the small area $i$ and block $q$.}
	
	The final expression for the estimator of the conditional MSE of $\hat{\bar{Y}}_i^{\star MQ}$ is the sum of \eqref{est_var_pred_MQ} and the square
	of \eqref{biasMQ}:
	\begin{equation}\label{MSE_MQ}
	\widehat{MSE}_{\bA,M}(\hat{\bar{Y}}_i^{\star MQ}|\hat\tau_i^{\star})=\hat{V}_{\bA, M}(\hat{\bar{Y}}_i^{\star MQ})+\hat{B}_{\bA,M}^2(\hat{\bar{Y}}_i^{\star MQ}).
	\end{equation}
	
	{\color{black}
		Following the approach of  \citet{Bianchi15}, we approximate the contribution to the MSE of (\ref{MQ_predictor_linked}) due to estimation of the area M-quantile coefficient $\tau_i^{\star}$ by 
		\begin{equation}
		V_{\bA,M}(\hat\tau_{i}^{\star})=\sum_{q \in i}\bar{\bx}_{iq}\bm{G}_{\tau_{i}^{\star}}^{\prime}\bm{G}_{\tau_{i}^{\star}}\bar{\bx}_{iq}^{\prime}v^2_{\hat\tau_{i}^{\star}},
		\end{equation}
		where $\bm{G}_{\tau_{i}^{\star}}=n^{-1}\sum_{h=1}^{D}\sum_{q \in h}\left(\bH^{-1}_{shq \tau_i^{\star}}\left \{\partial_{\tau_i^{\star}} \bL_{shq \tau_i^{\star}}-  \partial_{\tau_i^{\star}} \bH_{shq \tau_i^{\star}} \bH_{shq \tau_i^{\star}}^{-1} \bL_{shq \tau_i^{\star}} \right \}   \right)$ with \linebreak $\bH_{shq \tau_i^{\star}}=\bX_{shq}^{\star\prime}\bUpsilon_{shq \tau_i^{\star}}^{-1/2}\bW_{shq \tau_i^{\star}}^{\star} \bUpsilon_{shq \tau_i^{\star}}^{-1/2}\bX_{shq}^{\star}$,  $\bL_{shq \tau_i^{\star}}=\bX_{shq}^{\star \prime}\bUpsilon_{shq \tau_i^{\star}}^{-1/2}\bW_{shq \tau_i^{\star}}^{\star} \bUpsilon_{shq \tau_i^{\star}}^{-1/2} \by_{shq}^{\star}$, \linebreak $\partial_{\tau_i^{\star}} \bH_{shq \tau_i^{\star}}=\bX_{shq}^{\star\prime}\bUpsilon_{shq \tau_i^{\star}}^{-1/2}\partial_{\tau_i^{\star}}\bW_{shq \tau_i^{\star}}^{\star} \bUpsilon_{shq \tau_i^{\star}}^{-1/2}\bX_{shq}^{\star}$, $\partial_{\tau_i^{\star}}\bL_{shq  \tau_i^{\star}}=\bX_{shq}^{\star \prime}\bUpsilon_{shq \tau_i^{\star}}^{-1/2}\partial_{\tau_i^{\star}}\bW_{shq\tau_i^{\star}}^{\star} \bUpsilon_{shq \tau_i^{\star}}^{-1/2} \by_{shq}^{\star}$, $\partial_{\tau_i^{\star}}\bW_{shq\tau_i^{\star}}^{\star} =2 \bUpsilon_{shq\tau_i^{\star}}^{1/2}\Big|\psi \left \{ \bUpsilon_{shq\tau_i^{\star}}^{-1/2}(\by_{shq}^{\star}-\bX_{shq}^{\star\prime}\bbeta_{\tau_{i}^{\star}}^{\star})  \right\}\Big|\left \{(\by_{shq}^{\star}-\bX_{shq}^{\star\prime}\bbeta_{\tau_{i}^{\star}}^{\star})  \right\}^{-1}$ and \linebreak $v^2_{\hat\tau_{i}^{\star}}=n_{i}^{-1}\sum_{j=1}^{n_i}(\hat{\tau}_{ij}^{\star}-\hat{\tau}_i^{\star})^2$. This expression has the plug-in estimator
		
		\[
		\hat V_{\bA,M}(\hat\tau_{i}^{\star})=\sum_{q \in i}\bar{\bx}_{iq}\hat{\bm{G}}_{\hat \tau_{i}^{\star}}^{\prime}\hat{\bm{G}}_{\hat\tau_{i}^{\star}} \bar{\bx}_{iq}^{\prime}\hat{v}^2_{\hat\tau_{i}^{\star}}.
		\]
		The final form of the MSE estimator of $\hat{\bar{Y}}_i^{\star MQ}$ is then
		\begin{equation}\label{est_mse_est_mq}
		\widehat{MSE}_{\bA,M}(\hat{\bar{Y}}_i^{\star MQ})=\hat{V}_{\bA, M}(\hat{\bar{Y}}_i^{\star MQ})+\hat{B}_{\bA,M}^2(\hat{\bar{Y}}_i^{\star MQ})+\hat V_{\bA,M}(\hat\tau_{i}^{\star}).
		\end{equation}
	}

	\section{Model-based simulations}\label{mod:based}
	We used model-based simulations of the various linked data based small area predictors described in Sections \ref{sec:eblup_star}, \ref{sec:REBLUP} and \ref{sec:MQ}, as well as their corresponding MSE estimators, to illustrate their performance in situations where there are both linkage errors as well as actual population outliers. The synthetic populations underpinning these simulations are based on those used by \citet{Cha14} with some modifications.
	
	Values for $y$ are generated from the equation $y_{ij} =100+5x_{ij} +u_i+\epsilon_{ij}$ $j=1,\dots,N_i$, $i=1,\dots,D$; values for $x$ are generated independently from a lognormal distribution with a mean of $1.0$ and a standard deviation of $0.5$ on the log-scale. The population is divided in 40 areas ($i=1,\dots,D=40$) each of size $N_i=100$. The random components $u_i$ and $\epsilon_{ij}$ are generated independently according to two scenarios:
	\begin{itemize}
		\item[(0,0)] $u\sim N(0,3)$ and $\epsilon \sim N(0,6)$. In this scenario there are artificial outliers caused by linkage errors.
		\item[(e,u)] $u\sim N(0,3)$ for areas $1-36$, $u\sim N(0,20)$ for areas $37-40$ and $\epsilon \sim \delta N(0,6)+(1-\delta)N(0,150)$ where $\delta$ is an independently generated Bernoulli random variable with $Pr(\delta=1)=0.97$, i.e. the individual effects are independent draws from a mixture of two normal distributions, with $97\%$ on average drawn from a `well-behaved' $N(0,6)$ distribution and $3\%$ on average drawn from an outlier $N(0, 150)$ distribution.
	\end{itemize}
	
	Linked data pairs $(x_{ij}; y_{ij}^{\star})$ are then generated using the exchangeable linkage errors model (\ref{lem}) with correct linkage probabilities $\lambda_q = 1.0, 0.9, 0.6$ and $0.4$ for blocks 1, 2, 3 and 4, respectively. In each area there are 25 units for each block, assigned randomly. In scenario $(e,u)$ there are both artificial and real outliers.
	
	Samples of size $n_i=5$ are selected by simple random sampling without replacement within each area.  As the blocks are not considered in the sampling design, most area-specific samples do not include units from every block. 
	
	Each scenario is independently simulated $1,000$ times. For each simulation the population values are generated according to the underlying scenario, a sample is selected in each area and the sample data are then used to compute estimates of each of the actual area means for $y$.
	Six different estimators are used for this purpose: the standard EBLUP, $\hat{\bar{Y}}_i^{EBLUP}$ \citep{RaoMol15}, which serves as a reference, its corrected version in case of linkage error $\hat{\bar{Y}}_i^{\star EBLUP}$, the REBLUP estimator of \citet{Sin09}, $\hat{\bar{Y}}_i^{REBLUP}$, equation \eqref{REBLUP}, and its corrected version $\hat{\bar{Y}}_i^{\star REBLUP}$, expression \eqref{REBLUP_star}, the estimator based on M-quantile regression model $\hat{\bar{Y}}_i^{MQ}$ \eqref{MQ_predictor} and its corrected version with linked data $\hat{\bar{Y}}_i^{\star MQ}$ \eqref{MQ_predictor_linked}. In all cases the influence function $\psi$ is a Huber-type function with tuning constant $c=1.345$. For  $\hat{\bar{Y}}_i^{\star EBLUP}$ we consider two different methods for estimating the area-specific random effects: prediction of the random effects as in expression \eqref{eblup_res_1} and prediction of the random effects neglecting the second addend in \eqref{eblup_res_2}. We indicate this alternative estimator as $\hat{\bar{Y}}_i^{\star \star EBLUP}$.
	
	For each estimator and for each small area, we computed the Monte Carlo estimate of the percentage of relative bias and the percentage of Relative Root MSE (RRMSE) and the corresponding efficiency. The relative bias of an estimator $\hat{\bar{Y}}_i$ for the actual mean $\bar{Y}_i$ of area $i$ is the average across simulations of the errors $\hat{\bar{Y}}_i-\bar{Y}_i$ divided by the corresponding average value of $\bar{Y}_i$, its RRMSE is the square root of the average across simulations of the squares of these errors, again divided by the average value of $\bar{Y}_i$, and its efficiency (EFF) is the value of the ratios of the actual MSE of each predictors to the actual MSE of the corresponding EBLUP. Table \ref{table_scen1} shows median values for these performance measures for the various simulation scenarios and estimators.
	
	The results set out in Table \ref{table_scen1} confirm our expectations regarding the behaviour of the corrected predictors based on the linked data. They show smaller bias and higher efficiency than the traditional small area estimators in both scenarios. The estimator $\hat{\bar{Y}}_i^{\star \star EBLUP}$ is best in terms of bias, whereas $\hat{\bar{Y}}_i^{\star MQ}$ recorded the lowest values of RRMSE, with $\hat{\bar{Y}}_i^{\star REBLUP}$ performing similarly. We see that claims in the literature \citep{RaoMol15, Cha06} about the superior outlier robustness of REBLUP and M-quantile-based estimators compared with the EBLUP certainly hold true with both artificial outliers - the $(0,0)$ case - and with artificial as well as real outliers - the $(e,u)$ case.
	
	\begin{table}
		\begin{center}\begin{tabular}{lrrrr}\hline 
				Predictor & \multicolumn{4}{l}{Results ($\%$) for the following scenarios:} \\ 
				&\multicolumn{2}{c}{Relative bias} & \multicolumn{2}{c}{RRMSE (EFF)}\\
				& $(0,0)$ & $(e,u)$ & $(0,0)$ &  $(e,u)$ \\\hline
				$\hat{\bar{Y}}_i^{EBLUP}$& 0.00 &  -0.01& 1.37 (100.0) &  1.42 (100.0)\\
				$\hat{\bar{Y}}_i^{\star EBLUP}$& 0.03 & 0.01 & 1.26 (91.9) &  1.28 (90.2)  \\
				$\hat{\bar{Y}}_i^{\star \star EBLUP}$ & 0.00 & 0.00 & 1.29 (94.6) &  1.36 (95.3)\\
				$\hat{\bar{Y}}_i^{REBLUP}$ & -0.06 & -0.06& 1.16 (84.3) &  1.20 (83.7) \\
				$\hat{\bar{Y}}_i^{\star REBLUP}$& -0.09 & -0.09 & 1.14 (82.7)&  1.17 (81.2) \\
				$\hat{\bar{Y}}_i^{MQ}$  & -0.19 &-0.17 & 1.31 (94.8) & 1.34 (92.9)\\
				$\hat{\bar{Y}}_i^{\star MQ}$ &-0.04 &-0.05 & 1.12 (81.4) & 1.17 (80.8)\\ \hline
			\end{tabular} \caption{\label{table_scen1} Model-based simulation results: median values of the percentage of relative bias and RRMSE of predictors of small area means with $n_i=5$. In parenthesis the values of the efficiency (EFF).}
		\end{center}
	\end{table}
	
	The performances of the MSE estimators for the EBLUP, REBLUP and M-quantile-based predictors are evaluated in Table \ref{table_scen2}. Here we are mainly interested in the performances of the MSE estimators for the corrected predictors based on the linked data. We assess the MSE estimators of $\hat{\bar{Y}}_i^{\star EBLUP}$ and $\hat{\bar{Y}}_i^{\star \star EBLUP}$, implemented via a modification of the unconditional MSE estimator \eqref{mse_eblup} following \citet{Pra90}. MSE estimation of $\hat{\bar{Y}}_i^{\star REBLUP}$ uses the linearization-based MSE estimator \eqref{est_mse_u_star},  presented in Section  \ref{mse_REBLUP}. For $\hat{\bar{Y}}_i^{\star MQ}$ the MSE estimator \eqref{est_mse_est_mq} of Section \ref{sec:mse:mq} is evaluated. We also consider the performances of the MSE estimators for $\hat{\bar{Y}}_i^{EBLUP}$ \citep{Pra90},  $\hat{\bar{Y}}_i^{REBLUP}$ and $\hat{\bar{Y}}_i^{MQ}$ \citep{Cha14} under both scenarios. Table \ref{table_scen2} shows the medians of the area-specific percentage relative bias and percentage RRMSE of these MSE estimators. We see that the MSE estimator for $\hat{\bar{Y}}_i^{\star \star EBLUP}$ tends to be somewhat biased low under the $(0,0)$ scenario. The MSE estimator for $\hat{\bar{Y}}_i^{\star MQ}$ is less biased than the corresponding estimator for $\hat{\bar{Y}}_i^{\star REBLUP}$ for this scenario. In general, the proposed MSE estimators work well under linkage errors in scenario $(0,0)$. Under the $(e,u)$ scenario the MSE estimators of $\hat{\bar{Y}}_i^{\star EBLUP}$ and $\hat{\bar{Y}}_i^{\star \star EBLUP}$ and the MSE estimator of $\hat{\bar{Y}}_i^{\star REBLUP}$ all tend to overestimate the actual MSE, whereas the MSE estimator for $\hat{\bar{Y}}_i^{\star MQ}$ is slightly negatively biased. We also see that the MSE estimators of $\hat{\bar{Y}}_i^{EBLUP}$, $\hat{\bar{Y}}_i^{\star EBLUP}$ and $\hat{\bar{Y}}_i^{\star \star EBLUP}$ are generally more stable than those for the REBLUP and M-quantile-based predictors.
	

	\begin{table}
		\begin{center}\begin{tabular}{lrrrr}\hline 
				Predictor &\multicolumn{4}{c}{Results ($\%$) for the following scenarios:} \\ 
				&\multicolumn{2}{c}{Relative bias} & \multicolumn{2}{c}{RRMSE}\\
				& $(0,0)$ & $(e,u)$ & $(0,0)$ &  $(e,u)$ \\\hline
				$\widehat{MSE}(\hat{\bar{Y}}_i^{EBLUP})$& -3.9 & 2.8 &22.5 &32.1  \\
				$\widehat{MSE}_{\bA,M}(\hat{\bar{Y}}_i^{\star EBLUP})$& 0.7 &14.1 & 21.5& 35.3  \\
				$\widehat{MSE}_{\bA, M}(\hat{\bar{Y}}_i^{\star \star EBLUP})$&  -3.1 & 6.9  & 21.8  & 30.8  \\
				$\widehat{MSE}(\hat{\bar{Y}}_i^{REBLUP})$ & 9.1 &6.1 & 52.8 & 44.4 \\
				$\widehat{MSE}_{\bA, M|\bu}(\hat{\bar{Y}}_i^{\star REBLUP})$ &-0.4 & 19.2&51.1 & 63.0\\
				$\widehat{MSE}(\hat{\bar{Y}}_i^{MQ})$    &-10.3 & -12.6& 57.6 & 56.9\\
				$\widehat{MSE}_{\bA, M}(\hat{\bar{Y}}_i^{\star MQ})$ & -2.7 & -4.7 & 43.1& 44.1\\ \hline
			\end{tabular} \caption{\label{table_scen2} Median values of percentage of relative bias and RRMSE of RMSE estimators in model-based simulation experiments.}
		\end{center}
	\end{table}
	
	\section{Design-based simulation}\label{sec:db}
	Design-based simulations complement model-based simulations for robust SAE since they allow us
	to evaluate the performance of SAE methods in the context of a real population and realistic
	sampling methods where we do not know the precise source of outlier contamination. From a finite
	population perspective we believe that this type of simulation constitutes a more practical and
	appropriate representation of SAE performance.
	
	The synthetic population underpinning the design-based simulation is based on the simulation experiment reported in \cite{Briscolini2018}; it comes from the European Statistical System Data Integration project \citep[ESSnet,][]{ESSnet} and from the Survey on Household Income and Wealth, Bank of Italy (SHIW), whose data are freely available in anonymous form. Specifically, the synthetic ESSnet population contains information on over $26,000$ individuals including \textit{name}, \textit{surname}, \textit{gender} and \textit{date of birth}. Two new variables have been added to the original data set: the \textit{annual income} and the \textit{domain} indicator. The latter comprises 18 areas resulting from  aggregation of  Italian administrative regions.
	
	Following  \cite{Briscolini2018}, we carry out a realistic record linkage and SAE simulation experiment by perturbing the ESSnet data set via the introduction of missing values and typos in some potential linking variables (name, surname, gender and date of birth). Moreover, for the purposes of the simulation study, annual income has been removed from the perturbed data set and the corresponding value of \textit{annual consumption} obtained from the SHIW survey has been added. 
	
	The classical version of the probabilistic record linkage model by \citet{Fellegi69} and \citet{Jaro89} has been implemented by using the function \texttt{compare.linkage} of the package \texttt{RecordLinkage} in  \texttt{R} \citep{recordlink} to link the perturbed data set with the original register population by using  \textit{surname} as key-variable and  \textit{age}, grouped in four categories, and \textit{domain} as blocking variables. The domain indicator has been used as block in the linkage process to guarantee the assumption that both registers include an area identifier measured without error (see Section \ref{sec:ass}). After the linkage process the proportion of correct links for the four categories of age are $\lambda_q=(0.86, 0.93, 0.88, 0.91)$.
	
	The aim of the design-based simulation is to compare the performance of different estimators, and their MSE estimators, of the mean consumption in each domain under repeated sampling from a fixed population using income as  the auxiliary variable. A total of $1,000$ independent random samples of size $n=268$ were then taken from the synthetic fixed population by randomly selecting units in the 18 domains, with sample sizes proportional to domain sizes unless the resulting size was less than 5, which was set as the minimum domain sample size. 
	
	Table \ref{table_db} shows the median percentage relative biases, the median percentage RRMSEs and the efficiency of the same estimators that were evaluated in Section \ref{mod:based} while Table \ref{table_db1} reports the median
	relative percentage biases and the median percentage RRMSEs of the corresponding estimators of the MSEs for
	these estimators. Here we see that estimators developed to allow for linkage error work well in terms of both bias and RRMSE compared with the unmodified  EBLUP, REBLUP and M-quantile-based predictors that ignore linkage error. The $\hat{\bar{Y}}_i^{\star EBLUP}$ and $\hat{\bar{Y}}_i^{\star \star EBLUP}$ estimators perform best in terms of bias, whereas $\hat{\bar{Y}}_i^{\star REBLUP}$ is best in terms of RRMSE. From these results we conclude that estimators that correct for linkage error seem to offer the most balanced performance in terms of both bias and MSE for this population.
	
	With reference to MSE estimation, Table \ref{table_db1} indicates that on average across areas the MSE estimator for $\hat{\bar{Y}}_i^{\star EBLUP}$ and $\hat{\bar{Y}}_i^{\star \star EBLUP}$ performs better than the MSE estimators of $\hat{\bar{Y}}_i^{\star REBLUP}$, $\hat{\bar{Y}}_i^{\star MQ}$ that are based on the linearization method. Furthermore, the MSE estimator for $\hat{\bar{Y}}_i^{\star REBLUP}$ improves on the MSE estimator for $\hat{\bar{Y}}_i^{\star MQ}$ in terms of efficiency. 
	
	The analysis in Table  \ref{table_db1} focuses on comparison of median estimation performance across areas. The relationship between the `true' (empirical) RMSE of each estimator and its estimator for each area is shown in Figure \ref{MSE_boxplot}, where boxplots illustrate the variability in the RMSE ratio, defined as the ratio of the average estimated RMSE for each area to the true RMSE. We can see that, as expected, the MSE estimators proposed for linkage error corrected estimators perform better than MSE estimators for the EBLUP, REBLUP and M-quantile-based predictors.

	\begin{table}
		\begin{center}\begin{tabular}{lrrr}\hline 
				Predictor&\multicolumn{1}{c}{Relative bias} & \multicolumn{1}{c}{RRMSE} & \multicolumn{1}{c}{EFF}\\ \hline
				$\hat{\bar{Y}}_i^{EBLUP}$ & -0.23 & 7.40&  100.0 \\
				$\hat{\bar{Y}}_i^{\star EBLUP}$& 0.44 &6.40  & 86.6  \\
				$\hat{\bar{Y}}_i^{\star \star EBLUP}$& 0.38 & 6.45 &  86.8 \\
				$\hat{\bar{Y}}_i^{REBLUP}$& -2.10 & 5.35 & 81.3 \\
				$\hat{\bar{Y}}_i^{\star REBLUP}$& -0.67 & 4.96 & 78.7 \\
				$\hat{\bar{Y}}_i^{MQ}$   & -3.88 & 7.47 & 101.4\\
				$\hat{\bar{Y}}_i^{\star MQ}$  & -2.63 & 6.72 & 86.7\\ \hline
			\end{tabular} \caption{\label{table_db} Design-based simulation results: median values of relative bias, RRMSE and efficiency (EFF) of predictors of small area means with $n=262$.}
		\end{center}
	\end{table}

	\begin{table}
		\begin{center}\begin{tabular}{lrr}\hline 
				Predictor  &\multicolumn{1}{c}{Relative bias} & \multicolumn{1}{c}{RRMSE}\\ \hline
				$\widehat{MSE}(\hat{\bar{Y}}_i^{EBLUP})$& 49.9 & 95.6  \\
				$\widehat{MSE}_{\bA,M}(\hat{\bar{Y}}_i^{\star EBLUP})$& -2.5 & 47.9  \\
				$\widehat{MSE}_{\bA,M}(\hat{\bar{Y}}_i^{\star \star EBLUP})$&  -3.3 & 48.0 \\
				$\widehat{MSE}(\hat{\bar{Y}}_i^{REBLUP})$&  12.2 & 50.1 \\
				$\widehat{MSE}_{\bA,M|\bu}(\hat{\bar{Y}}_i^{\star REBLUP})$& 12.0& 50.4\\
				$\widehat{MSE}(\hat{\bar{Y}}_i^{MQ})$   & -22.8& 51.3 \\
				$\widehat{MSE}_{\bA,M}(\hat{\bar{Y}}_i^{\star MQ})$  & 14.4 & 59.9\\ \hline
			\end{tabular} \caption{\label{table_db1} Median values of relative bias and RRMSE of RMSE estimators in design-based simulation experiment.}
		\end{center}
	\end{table}

	\begin{figure}[h]
		\centering    
		\makebox{\includegraphics[scale = 0.50]{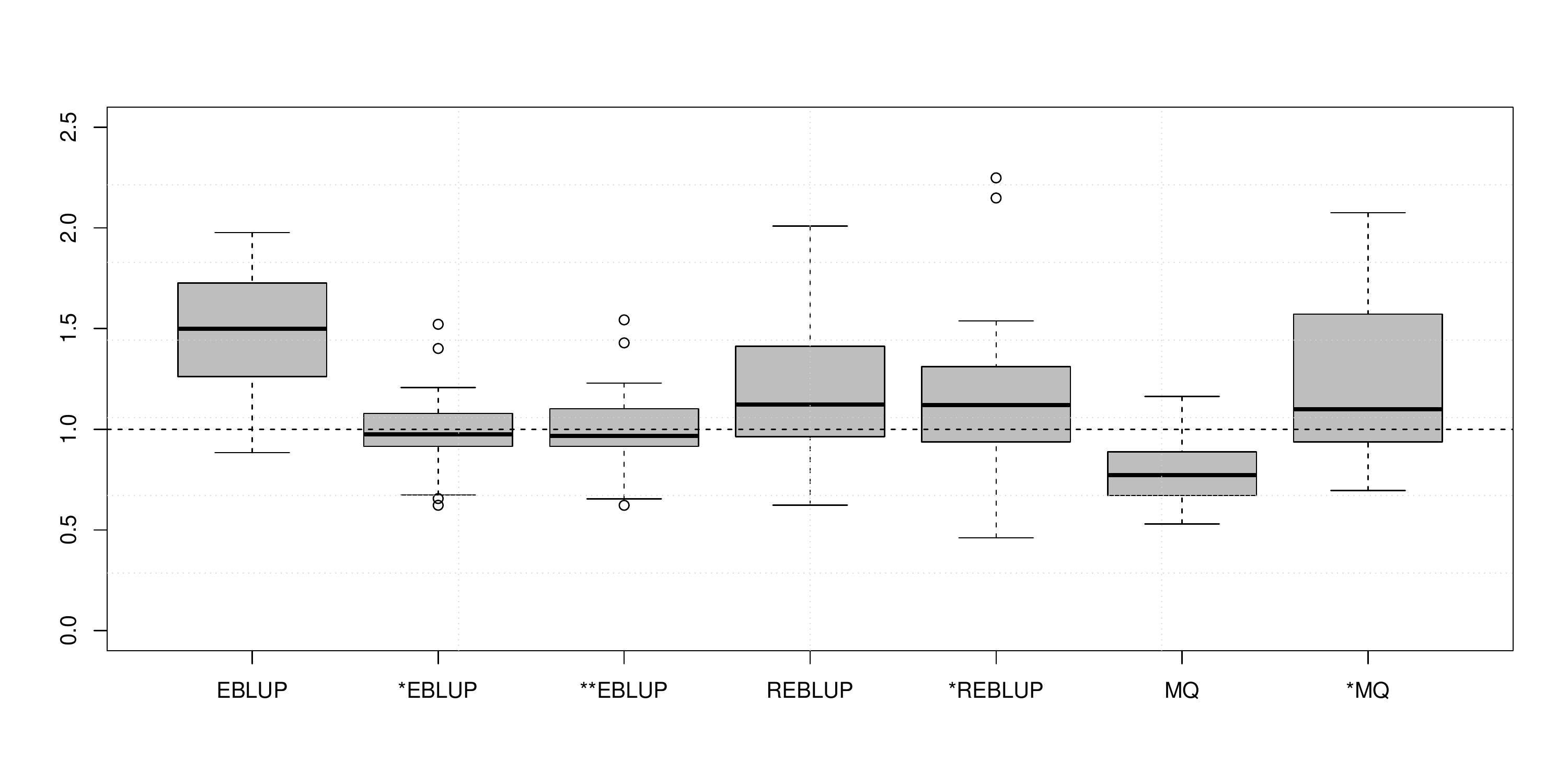} }
		\caption{\label{MSE_boxplot} Boxplots showing area-specific values of the RMSE ratios for the MSE estimators in the design-based
			scenario (the RMSE ratio is defined as the ratio of the average over repeated sampling of the RMSE estimator
			for an estimator to the actual RMSE of this estimator under repeated sampling).}
	\end{figure}

	\section{Final remarks}\label{sec:conclusions}
	
	In this paper, we propose a number of outlier-robust small area estimation methods that also allow for linkage error in the data. These proposed {\color{black}$\star$EBLUP, $\star$REBLUP and $\star$MQ-based estimators} have the potential to lead to significantly better small area estimates in important applications where linked data are available, such as in financial, economic, environmental and public health applications.
	
	The properties of the proposed estimators have been studied through model-based and design-based simulation studies. The results from these studies suggest
	that these estimators represent a promising alternative way to allow for linkage error in SAE. In particular, the empirical results reported in Sections \ref{mod:based} and \ref{sec:db} show that the proposed small area estimators are less biased and more efficient than the traditional  predictors in the presence of artificial (i.e. linkage error-induced) and real outliers. In addition, the performance of the proposed
	MSE estimators for these small area estimators seems promising, but we are aware that further research in this area is necessary. \texttt{R} code for calculating the {\color{black}$\star$EBLUP, $\star$REBLUP and $\star$MQ-based estimators} proposed in this paper and their corresponding MSE estimators is available from the authors upon request. 
	
	The approach to small area estimation using probability-linked data described above is in the spirit of \citet{Scheurenwinkler1993}, where it is suggested that one corrects the
	naive estimator using an estimate of its bias under an appropriate model for the linkage error process. In our case the adjustment we use for this purpose depends on assuming that linkage errors are generated via an ELE process and knowing the parameters (i.e. the $\lambda_q$) that characterise this process. This is highly unlikely to be the case, and these probabilities $\lambda_q$ will usually be estimated in some way. One way to estimate these parameters, suggested in \citet{OSR2009}, is via access to a random `audit' sample of the linked records in each block, where the only thing one needs to know is whether a sampled link is correct or not. This could also be accomplished by calculating the achieved linkage error rate in a training set of `gold standard' links, as would be possible if a classification-based approach to linkage was used. In general, we can think of these estimated probabilities as part of the paradata for the linkage process, which should be made available to users of the linked data. These estimates can be substituted into the expressions for the proposed small area estimators in Sections \ref{sec:eblup_star}, \ref{sec:REBLUP} and \ref{sec:MQ}. In order to assess the performance of the proposed small area estimators in this case (i.e. when linkage error rates are estimated) we have replicated the model-based simulations of Section \ref{mod:based} with $\lambda_q$ estimated by independently selecting a random `audit' sample of linked records of 25 units in each block. The results in this case show a very small increase in the empirical variability of the proposed {\color{black}$\star$EBLUP, $\star$REBLUP and $\star$MQ-based estimators}. Interested readers can contact the authors to access these more detailed results. 
	
	The extra uncertainty arising from the estimation of the probabilities $\lambda_q$ needs to be accounted for when carrying out MSE estimation for the small area estimators that use $\bA_q$ to correct for bias induced by linkage errors. This extra uncertainty can be taken into account in the estimated MSE of $\hat{\bar{Y}}_i^{\star EBLUP}$ by adding a term $g_{4i}(\hat{\bdelta},\hat{\lambda}_q)$ to expression \eqref{mse_eblup}:
	\begin{equation}
	g_{4i}(\hat{\bdelta},\hat{\lambda}_q)  = tr \left[ \frac{\partial \bb_i ^{\prime}}{\partial \lambda_q} \sum_{q \in i}\hat \bSigma_{siq} \hat{V}(\hat{\lambda}_q)  \left( \frac{\partial \bb_i ^{\prime}}{\partial \lambda_q} \right)^{\prime}      \right],
	\end{equation}
	where $\hat{V}(\hat{\lambda}_q)$ is an estimator of the variance of the estimators of the probabilities of correct linkage. If the estimates of the linkage probabilities are obtained via an `audit' sample, $\hat{V}(\hat{\lambda}_q)=(\sum_{i=1}^D n_{iq})\hat{\lambda}_q(1-\hat{\lambda}_q)$. In which case the estimator of the MSE of $\hat{\bar{Y}}_i^{\star EBLUP}$ becomes
	\begin{equation}\label{mse_eblup_1}
	\widehat{MSE}_{\bA,M}(\hat{\bar{Y}}_i^{\star EBLUP})=(1-n_i/N_i)^2\{g_{1i}(\hat{\bdelta},\hat{\lambda}_q)+g_{2i}(\hat{\bdelta},\hat{\lambda}_q)+2g_{3i}(\hat{\bdelta},\hat{\lambda}_q)+g_{4i}(\hat{\bdelta},\hat{\lambda}_q)\}+o(D^{-1}).
	\end{equation}
	
	As far as estimation of the MSE of $\hat{\bar{Y}}_i^{\star REBLUP}$ defined by \eqref{REBLUP_star} is concerned, we note that the component $h_{1i}(\hat{\btheta})$ of \eqref{est_mse_u_star} now becomes
	{\color{black}
		\begin{equation}
		h_{1i}(\hat{\btheta})=\left(1-\frac{n_i}{N_i}\right)^2 \begin{array}{c|c|c}(\bar{\bx}_{ri}^{\star\prime}~~ & ~~\bar{\bz}_{ri}^\prime~~ &~~ \uno_{q}^{\prime})\end{array}\hat{V}_{\bA, M| \bu}(\hat\btheta,\hat{\bLambda})\begin{array}{c|c|c}(\bar{\bx}_{ri}^{\star\prime}~~ &~~ \bar{\bz}_{ri}^\prime ~~&~~  \uno_{q}^{\prime})\end{array}^{\prime},
		\end{equation}
	}
	where $\hat{\bLambda}$ denotes the vector defined by the block-specific values of $\lambda_q$ and $\hat{V}_{\bA, M| \bu}(\hat\btheta,\hat{\bLambda})$ is the estimated joint variance of $\hat\btheta$ and $\hat{\bLambda}$ obtained by computing the asymptotic variance of solutions to the estimating equations.
	Using the same approach, we note that the first component of the MSE estimator \eqref{est_var_pred_MQ} of the M-quantile based estimator \eqref{MQ_predictor_linked} also depends on the extra uncertainty arising from estimation of the probabilities of correct linkage and so needs to be written as
	{\color{black}
		$$(1-n_i/N_i)^2\left(\begin{array}{c|c}(\bar{\bx}_{ri}^{\star\prime} ~~& ~~\uno_{q}^{\prime})\end{array}\hat{V}_{\bA,M}(\hat{\bbeta}_{\tau_i^{\star}}^{\star},\hat{\bLambda}) \begin{array}{c|c}(\bar{\bx}_{ri}^{\star\prime}~~ &~~ \uno_{q}^{\prime})\end{array}^{\prime} \right).$$}
	The performance of the MSE estimators {\color{black}$\star$EBLUP, $\star$REBLUP \eqref{REBLUP_star} and the $\star$MQ-based} estimator \eqref{MQ_predictor_linked} when there is extra uncertainty arising from the estimation of probabilities of correct linkage is an area of current research. 
	
	
	Despite the fact that the proposed methodologies provide encouraging results, further research remains to be done. In this paper we have assumed that both registers include an area identifier measured without error that is used in the linkage process. Consequently we do not allow   units from different areas to be erroneously linked. When this assumption is relaxed, the incidence variable $\bZ_{siq}$ becomes $\bZ_{siq}^{\star}=\bA_{siq} \bZ_{siq}$, and the area effects associated with $\by_{s}^{\star}$ are correlated across areas. This correlation needs to be taken into when calculating {\color{black}$\star$EBLUP and $\star$REBLUP}. We conjecture that this correlation has less effect on small area estimators based on an M-quantile regression model because this type of model does not assume independent area effects. 
	
	Finally, we note once more that we have assumed a simple exchangeable linkage error model because we are focussed on SAE carried out by a secondary data analyst. It would be interesting to extend our estimators to situations where the information set on the linkage process is richer and a primary analyst viewpoint can be adopted as in \citet{Briscolini2018} and \citet{Lahiri18}.
	
	\vspace{9pt}
	\textbf{Acknowledgment}: The work of Nicola Salvati has been carried out with the support of the project InGRID 2 (Grant Agreement N.
	730998, EU) and of project PRA2018-9 (From survey-based to register-based statistics: a paradigm shift using latent variable models).
	
	\vspace{9pt}
	\section*{Appendix A}
	
	\setcounter{figure}{0} \setcounter{table}{0}
	\setcounter{equation}{0}
	\renewcommand{\theequation}{A-\arabic{equation}}
	\subsection*{Proof of \eqref{jointvar}}
	
	We start noting that
	\[
	V_{\bA,M}(\by^{\star}_{siq}|\bX_{iq})=V_ME_{\bA}(\by^{\star}_{siq}|\bX_{iq})+V_{\bA}E_M(\by^{\star}_{siq}|\bX_{iq})
	\]
	Now $E_{\bA}(\by^{\star}_{siq}|\bX_{iq})=\bT_{siq}\bX_{iq}\bbeta+\bZ_{siq}\bu_i+\be^\star_{siq}=\bX^\star_{siq}\bbeta+\bZ_{siq}\bu_i+\be^\star_{siq}$ and $E_M(\by^{\star}_{siq}|\bX_{iq})=\bA_{siq}\bX_{iq}\bbeta$. As a consequence:
	\begin{eqnarray}
	\nonumber V_{\bA,M}(\by^{\star}_{siq}|\bX_{iq}) &=& V_M (\bX^\star_{siq}\bbeta + \bZ_{siq}\bu_i+\be^\star_{siq}) + V_{\bA}(\bA_{siq}\bX_{iq}\bbeta) \\
	\nonumber &=& \bZ_{siq}\boldsymbol{\Sigma}_{\bu_i}\bZ^T_{siq} + \boldsymbol{\Sigma}_{seiq}+\bV_{siq}=\boldsymbol{\Sigma}_{siq}.
	\end{eqnarray}
	
	\subsection*{Proof of Preposition 1}
	In this Section we provide the proof of Preposition 1 paralleling \cite{RaoMol15}, Section 5.6.1. Assuming $\by_s^{\star}=(\by_{siq}^\star),~i=1,\dots,D;~q=1,\dots,Q$ and $\bX_s^{\star \prime}=col_{1\leqslant i,q \leqslant D,Q}(\bX_{siq}^{\star})^\prime$ a linear estimator $\hat\mu=\ba^\prime \by_s^\star+b$ is unbiased for $\mu=\bl^\prime \bbeta+\bc \bu$ under the linear mixed model \eqref{sample_mod}, that is $E(\hat\mu)=E(\mu)$, if and only if $\ba^\prime\bX_s^\star=\bl^\prime$ and $b=0$. The $MSE$ of $\hat{\mu}$ is given by
	\begin{eqnarray}
	\nonumber MSE_{\bA,M}(\hat{\mu})& = & V_{\bA,M}(\hat{\mu}-\mu)=V_{\bA,M}(\ba^{\prime}\by_{s}^{\star}-\bl^\prime \bbeta-\bc \bu)\\
	\nonumber& = & \ba^{\prime}V_{\bA,M}(\by_{s}^{\star})\ba+\bc^{\prime}V_{\bA,M}(\bu)\bc-2\ba^{\prime}Cov_{\bA,M}(\by_{s}^{\star},\bu)\bc\\
	\nonumber& = & \ba^{\prime}\left(\bZ\bSigma_\bu \bZ^\prime+\bSigma_{se}+\bV \right)\ba+\bc^\prime \bSigma_\bu \bc-2\ba^\prime\bZ\bSigma_\bu\bc\\
	& = & \ba^{\prime}\bSigma_s \ba+\bc^\prime \bSigma_\bu \bc-2\ba^\prime\bZ\bSigma_\bu\bc,
	\end{eqnarray}
	where $\bV=V_\bA(\bA_s \bX \bbeta)$ with $\bA_s^{\prime}=col_{1\leqslant i,q \leqslant D,Q}(\bA_{siq})^\prime$ and $\bSigma_{se}$ is the variance of $\be_s^{\star}=(\be_{siq}^\star),~i=1,\dots,D;~q=1,\dots,Q$. The BLUP estimator is obtaining minimising $MSE_{\bA,M}(\hat{\mu})$ subject to the unbiasedness constraint $\ba^\prime\bX_s^\star=\bl^\prime$. Formulas \eqref{wls} and \eqref{rablup} can be obtained following Section 5.6.1 from \cite{RaoMol15} with simple changing in notation.
	
	REMARK: As $\bSigma_s$ is block-diagonal matrix, $\bSigma_s=diag(\bSigma_{siq}),~i=1,\dots,D;~q=1,\dots,Q$ equations \eqref{wls} and \eqref{rablup} are expressed using summations, that is in a form more efficient for computation.
	
	\section*{Appendix B}\label{appendixb}
	The following RCs are required to prove the Preposition 2 set out in Section \ref{sec:MSE:EBLUP} and Prepositions 3 and 4 of Section \ref{mse_REBLUP} and uses the same notation as employed there. The regularity conditions are similar to those proposed by \citet{Pra90} and \citet{Cha14} with some differences due to the presence of the linkage error.
	\begin{itemize}
		\item[] \textit{Condition 1}. The elements of $\sum_{i,q}\bX_{siq}^{\star}$ are uniformly bounded as $D \rightarrow \infty$ such that \linebreak $\sum_{i=1}^{D} \sum_{q \in i} \bX_{siq}^{\star\prime} {\bSigma}_{siq}^{-1} \bX_{siq}^{\star}=[O(D^{-1})]_{p \times p}$.
		\item[] \textit{Condition 2}. The covariances matrices $\bSigma_{\bu_i}$ and  $\bSigma_{siq}$ $i=1,\cdots,D$ and $q=1,\cdots,Q$ have linear structure and are known positive definite matrices of order $m \times m$ and $n_{iq} \times n_{iq}$ respecitvely, with elements that are also uniformly bounded as $D \rightarrow \infty$.
		\item[] \textit{Condition 3}. The elements of the matrix $\bV_s=diag(\bV_{siq})$ are uniformly bounded as $D \rightarrow \infty$.
		\item[] \textit{Condition 4}. The covariances matrices $\bSigma_{\bu_i}$, $\bSigma_{siq}$ and $\bU_{siq}$, $i=1,\cdots,D$ and $q=1,\cdots,Q$ are differentiable with respect to the variance components.
		\item[] \textit{Condition 5}. The dimension $m$ of the area random effect is a fixed finite number $\sup_{i \geqslant 1}n_i=\Delta <\infty$.
		\item[] \textit{Condition 6}. $\hat\bdelta$ is a translation-invariant unbiased estimator of $\bdelta$ as in \textit{Condition 4} of \citet{Pra90}.
		\item[] \textit{Condition 7}. The influence function $\psi$ is a bounded continuous function with a derivative which, except for a finite number of points, is defined everywhere and is also bounded as in \textit{Condition 1} of \citet{Cha14}.
		\item[] \textit{Condition 8}. There are constants $\zeta>0$ and $L < \infty$ such that, if $\br_{siq}^{\star}=\bU_{siq}^{-1/2}(\by_{siq}^{\star}-\bX_{siq}^{\star}\bbeta_0^{\psi\star})$, $\bt_{siq}^{\star}=(\bSigma_{seiq}+\bV_{siq})^{-1/2} (\by_{siq}^{\star}-\bX_{siq}^{\star}  \bbeta_0^{\psi\star} -\bZ_{siq} \bu_{0i}^{\psi \star})$ and $\bd=\bSigma_{\bu_i}^{-1/2}\bu_{0i}^{\psi \star}$ then $E_{\bA, M|\bu}|  \psi\{\br_{siq}^{\star}  \}  |^{4+\zeta}$, $E_{\bA, M|\bu}|| \partial \psi\{\br_{siq}^{\star}  \}  ||$, $E_{\bA, M|\bu}|  \psi\{\bt_{siq}^{\star}  \}  |^{4+\zeta}$, $E_{\bA, M| \bu}|| \partial \psi\{\bt_{siq}^{\star}  \}  ||$, $E_{\bA, M| \bu}|  \psi\{\bd_{siq}^{\star}  \}  |^{4+\zeta}$ and $E_{\bA, M|\bu}|| \partial \psi\{\bd_{siq}^{\star}  \}  ||$ are all bounded by $L$ as in \textit{Condition 5} in \citet{Cha14}.
		\item[] \textit{Condition 9}. $\partial_{\delta_k}\bX_{siq}^{\star} \bB_{siq}=[O(1)]_{p \times m}$  for $k=1, \dots, K$ where $\bB_{siq}$ was defined in equation \eqref{BSIQ}. This follows \textit{Condition 6} in \citet{Cha14}.
	\end{itemize}
	
	\section*{Appendix C}
	
	In this Appendix we report $\hat{V}_{\bA, M| \bu}(\tilde\theta)$ that is a sandwich-type estimator of the first order approximation of the conditional covariance matrix $V_{\bA, M| \bu}(\tilde\theta)$ defined in \cite{Cha14}.  From equations \eqref{rob1_star} and \eqref{rra} we note that $\bH(\tilde\theta)=\bzero$ where
	{\scriptsize \begin{equation*}
		\bH(\btheta)=\left(\begin{array}{l} \bH_{\bbeta^{\psi \star}} \\ \bH_{\bu^{\psi \star}}\end{array}\right)=\left(\begin{array}{l}  \sum_{i=1}^D \sum_{q \in i}\bX_{siq}^{\star \prime}\bSigma_{siq}^{-1}\bU_{siq}^{1/2}\psi\{\bU_{siq}^{-1/2}(\by_{siq}^{\star}-\bX_{siq}^{\star}\bbeta^{\psi \star})\}= \bzero\\      \sum_{i=1}^D \sum_{q \in i} \left \{\bZ_{siq}^\prime \bSigma_{seAiq}^{-1/2} \psi\left\{\bSigma_{seAiq}^{-1/2} (\by_{siq}^{\star}-\bX_{siq}^{\star}  \bbeta^{\psi \star} -\bZ_{siq} \bu_{i}^{\psi \star}) \right\}- \bSigma_{\bu_i}^{ -1/2} \psi\{ \bSigma_{\bu_i}^{-1/2}\bu_i^{\psi \star}  \} \right\} =\bzero        \end{array}\right).
		\end{equation*}}
	Then we can compute the asymptotic variance of solutions to an estimating equation to obtain a first order approximation to $V_{\bA, M| \bu}(\tilde \theta)$ following \cite{Cha14}. The sandwich-type estimator of this asymptotic approximation can be obtained as
	\begin{equation}\label{sandwich-type}
	\hat{V}_{\bA, M| \bu}(\tilde\theta)=\hat{E}_{\bA, M| \bu}(\partial_{\btheta} \bH_{0})^{-1}\left(\begin{array}{cc}  \hat{V}_{\bA, M| \bu} (\bH_{0\bbeta^{\psi\star}})& \widehat{Cov}_{\bA, M| \bu} (\bH_{0\bbeta^{\psi\star}},\bH_{0\bu^{\psi\star}})\\ \widehat{Cov}_{\bA, M| \bu} (\bH_{0\bu^{\psi\star}},\bH_{0\bbeta^{\psi\star}}) &  \hat{V}_{\bA, M| \bu} (\bH_{0\bu^{\psi\star}}) \end{array}\right)\left(\hat{E}_{\bA, M| \bu}(\partial_{\btheta} \bH_{0})^{-1}\right)^{\prime},
	\end{equation}
	where
	\begin{equation*}
	\hat{E}_{\bA, M| \bu}(\partial_{\btheta} \bH_{0})^{-1}=\left( \begin{array}{cc}  \hat{E}_{\bA, M| \bu}(\partial_{\bbeta_0^{\psi \star}}\bH_{0 \bbeta^{\psi \star}}) & - \hat{E}_{\bA, M| \bu}(\partial_{\bbeta_0^{\psi \star}}\bH_{0 \bbeta^{\psi \star}})^{-1} \hat{E}_{\bA, M| \bu}(\partial_{\bbeta_0^{\psi \star}}\bH_{0 \bu^{\psi \star}}) \hat{E}_{\bA, M| \bu}(\partial_{\bu^{\psi \star}}\bH_{0 \bu^{\psi \star}}) ^{-1}\\ \bzero &  \hat{E}_{\bA, M| \bu}(\partial_{\bu_0^{\psi \star}}\bH_{0 \bu_0^{\psi \star}}) \end{array}   \right),
	\end{equation*}
	with
	\begin{equation*}
	\hat{E}_{\bA, M| \bu}(\partial_{\bbeta_0^{\psi \star}}\bH_{0 \bbeta^{\psi \star}}) =- \sum_{i=1}^D \sum_{q \in i}\bX_{siq}^{\star \prime}\bSigma_{siq}^{-1}\bU_{siq}^{1/2}\bR_{siq}^{*}\bU_{siq}^{-1/2}\bX_{siq}^{\star },
	\end{equation*}
	
	\begin{equation*}
	\hat{E}_{\bA, M| \bu}(\partial_{\bu_0^{\psi \star}}\bH_{0 \bu_0^{\psi \star}})=-\sum_{i=1}^D \sum_{q \in i} \bZ_{siq}^\prime \bSigma_{seAiq}^{-1/2}\bT_{siq}^{\star} \bSigma_{seAiq}^{-1/2}\bZ_{siq}-\bSigma_{\bu_i}^{ -1/2} \bC_{siq}^{\star}\bSigma_{\bu_i}^{ -1/2},
	\end{equation*}
	
	\begin{equation*}
	\hat{E}_{\bA, M| \bu}(\partial_{\bbeta_0^{\psi \star}}\bH_{0 \bu_0^{\psi \star}})=-\sum_{i=1}^D \sum_{q \in i} \bZ_{siq}^\prime \bSigma_{seAiq}^{-1/2}\bT_{siq}^{\star} \bSigma_{seAiq}^{-1/2}\bX_{siq}^{\star},
	\end{equation*}
	
	\begin{equation*}
	\hat{V}_{\bA, M| \bu} (\bH_{0\bbeta^{\psi\star}})=(n-p)^{-1} \sum_{i=1}^D \sum_{q \in i}\psi^{2}\{\br_{siq}^{\star}\}\bX_{siq}^{\star \prime} \bSigma_{siq}^{-1}\bU_{siq}\bSigma_{siq}^{-1}\bX_{siq}^{\star}
	\end{equation*}
	
	\begin{equation*}
	\hat{V}_{\bA, M| \bu} (\bH_{0\bu^{\psi\star}})=(n-p)^{-1} \sum_{i=1}^D \sum_{q \in i}\psi^{2}\{\bt_{siq}^{\star}\}\bZ_{siq}^{\prime}\bSigma_{seAiq}^{-1}\bZ_{siq}
	\end{equation*}
	
	\begin{equation*}
	\widehat{Cov}_{\bA, M| \bu} (\bH_{0\bu^{\psi\star}},\bH_{0\bbeta^{\psi\star}}) =(n-p)^{-1}\sum_{i=1}^D \sum_{q \in i}\left(\psi\{\br_{siq}^{\star}\}\psi\{\bt_{siq}^{\star}\}\right)\bX_{siq}^{\star \prime} \bSigma_{siq}^{-1}\bU_{siq}^{1/2}\bSigma_{seAiq}^{-1/2}\bZ_{siq},
	\end{equation*}
	with $\bR_{siq}^{\star}$ is an $n_{iq} \times n_{iq}$ diagonal matrix of $\br_{siq}^{\star}=\bU_{siq}^{-1/2}(\by_{siq}^{\star}-\bX_{siq}^{\star}\tilde\bbeta^{\star})$ with $j$th diagonal element equal to 1 if $-c < r_{siqj}^{\star} <c$, and $0$ otherwise; $\bT_{siq}^{\star}$ is an $n_{iq} \times n_{iq}$ diagonal matrix of $\bt_{siq}^{*}=\bSigma_{seAiq}^{-1/2} (\by_{siq}^{\star}-\bX_{siq}^{\star} \tilde \bbeta^{\star} -\bZ_{siq} \tilde\bu_{i}) $ with $j$th diagonal element equal to 1 if $-c < t_{siqj}^{\star} <c$, and $0$ otherwise; $\bC_{siq}^{\star}$ is an $m \times m $ diagonal matrix of $ \bSigma_{\bu_i}^{-1/2}\tilde\bu_i^{\psi \star} $with $l$th diagonal element equal to 1 if $-c< c_{siqh}^{\star} <c$, and $0$ otherwise.
	\bibliographystyle{chicago} 
	\bibliography{Datenbank}  

\begin{thebibliography}{}

\bibitem[\protect\citeauthoryear{Battese, Harter, and Fuller}{Battese
  et~al.}{1988}]{Bat88}
Battese, G.~E., R.~M. Harter, and W.~A. Fuller (1988).
\newblock An error component model for prediction of county crop areas using
  survey and satellite data.
\newblock {\em Journal of the American Statistical Association\/}~{\em 83
  (401)}, 28--36.

\bibitem[\protect\citeauthoryear{Bianchi, Fabrizi, Salvati, and
  Tzavidis}{Bianchi et~al.}{2018}]{Bianchi18}
Bianchi, A., E.~Fabrizi, N.~Salvati, and N.~Tzavidis (2018).
\newblock Estimation and testing in m-quantile regression with applications to
  small area estimation.
\newblock {\em International Statistical Review\/}~{\em 86}, 541--570.

\bibitem[\protect\citeauthoryear{Bianchi and Salvati}{Bianchi and
  Salvati}{2015}]{Bianchi15}
Bianchi, A. and N.~Salvati (2015).
\newblock Asymptotic properties and variance estimators of the m-quantile
  regression coefficients estimators.
\newblock {\em Communications in Statistics - Theory and Methods\/}~{\em 44},
  2416--2429.

\bibitem[\protect\citeauthoryear{Booth and Hobert}{Booth and
  Hobert}{1998}]{booth98}
Booth, J. and J.~Hobert (1998).
\newblock Standard errors of prediction in generalized linear mixed models.
\newblock {\em Journal of the American Statistical Association\/}~{\em 93},
  262--272.

\bibitem[\protect\citeauthoryear{Borg and Sariyar}{Borg and
  Sariyar}{2019}]{recordlink}
Borg, A. and M.~Sariyar (2019).
\newblock {\em RecordLinkage: Record Linkage in R}.
\newblock R package version 0.4-10.1.

\bibitem[\protect\citeauthoryear{Breckling and Chambers}{Breckling and
  Chambers}{1988}]{Bre88}
Breckling, J. and R.~Chambers (1988).
\newblock M-quantiles.
\newblock {\em Biometrika\/}~{\em 75 (4)}, 761--771.

\bibitem[\protect\citeauthoryear{Briscolini, Consiglio, Liseo, Tancredi, and
  Tuoto}{Briscolini et~al.}{2018}]{Briscolini2018}
Briscolini, D., L.~D. Consiglio, B.~Liseo, A.~Tancredi, and T.~Tuoto (2018).
\newblock New methods for small area estimation with linkage uncertainty.
\newblock {\em Interational Journal of Approximate Reasoning\/}~{\em 94},
  30--42.

\bibitem[\protect\citeauthoryear{Chambers}{Chambers}{1986}]{Cha86}
Chambers, R. (1986).
\newblock Outlier robust finite population estimation.
\newblock {\em Journal of the American Statistical Association\/}~{\em 81},
  1063--1069.

\bibitem[\protect\citeauthoryear{Chambers}{Chambers}{2009}]{OSR2009}
Chambers, R. (2009).
\newblock Regression analysis of probability-linked data.
\newblock Technical report, Official Statistics Research, Statistics New
  Zealand, Downloaded from
  http://www.statisphere.govt.nz/official-statistics-research/series/vol-4.htm.

\bibitem[\protect\citeauthoryear{Chambers, Chandra, Salvati, and
  Tzavidis}{Chambers et~al.}{2014}]{Cha14}
Chambers, R., H.~Chandra, N.~Salvati, and N.~Tzavidis (2014).
\newblock Outlier robust small area estimation.
\newblock {\em Journal of the Royal Statistical Society: Series B\/}~{\em 76
  (1)}, 47--69.

\bibitem[\protect\citeauthoryear{Chambers, Chandra, and Tzavidis}{Chambers
  et~al.}{2011}]{Cha11a}
Chambers, R., J.~Chandra, and N.~Tzavidis (2011).
\newblock On bias-robust mean squared error estimation for pseudo-linear small
  area estimators.
\newblock {\em Survey Methodology\/}~{\em 37 (2)}, 153--170.

\bibitem[\protect\citeauthoryear{Chambers and Tzavidis}{Chambers and
  Tzavidis}{2006}]{Cha06}
Chambers, R. and N.~Tzavidis (2006).
\newblock M-quantile models for small area estimation.
\newblock {\em Biometrika\/}~{\em 93 (2)}, 255--268.

\bibitem[\protect\citeauthoryear{Fellegi and Sunter}{Fellegi and
  Sunter}{1969}]{Fellegi69}
Fellegi, I. and A.~Sunter (1969).
\newblock A theory for record linkage.
\newblock {\em Journal of the American Statistical Association\/}~{\em 64},
  1183--1210.

\bibitem[\protect\citeauthoryear{Fellner}{Fellner}{1986}]{Fel86}
Fellner, W.~H. (1986).
\newblock Robust estimation of variance components.
\newblock {\em Technometrics\/}~{\em 28 (1)}, 51--60.

\bibitem[\protect\citeauthoryear{Han and Lahiri}{Han and
  Lahiri}{2018}]{Lahiri18}
Han, Y. and P.~Lahiri (2018).
\newblock Statistical analysis with linked data.
\newblock {\em International Statistical Review\/}~{\em
  doi:10.1111/insr.12295}, 1--19.

\bibitem[\protect\citeauthoryear{Harville and Jeske}{Harville and
  Jeske}{1992}]{Har92}
Harville, D.~A. and D.~R. Jeske (1992).
\newblock Mean square error of estimation or prediction under a general linear
  model.
\newblock {\em Journal of the American Statistical Association\/}~{\em 87},
  724--731.

\bibitem[\protect\citeauthoryear{Henderson}{Henderson}{1975}]{Hen75}
Henderson, C.~R. (1975).
\newblock Best linear unbiased estimation and prediction under a selection
  model.
\newblock {\em Biometrics\/}~{\em 31 (2)}, 423--447.

\bibitem[\protect\citeauthoryear{Huber}{Huber}{1981}]{Hub81}
Huber, P. (1981).
\newblock {\em Robust Statistics}.
\newblock New York: Wiley.

\bibitem[\protect\citeauthoryear{Jaro}{Jaro}{1989}]{Jaro89}
Jaro, M. (1989).
\newblock Advances in record-linkage methodology as applied to matching the
  1985 census of tampa, florida.
\newblock {\em Journal of the American Statistical Association\/}~{\em 84},
  414--420.

\bibitem[\protect\citeauthoryear{Kim and Chambers}{Kim and
  Chambers}{2012}]{Kimchambers2012}
Kim, G. and R.~Chambers (2012).
\newblock Regression analysis under incomplete linkage.
\newblock {\em Computational Statistics and Data Analysis\/}~{\em 56},
  2756--2770.

\bibitem[\protect\citeauthoryear{Kim and Chambers}{Kim and
  Chambers}{2015}]{KimCha2015}
Kim, G. and R.~Chambers (2015).
\newblock Unbiased estimation in the presence of correlated linkage error.
\newblock {\em Stat\/}~{\em 4}, 32--45.

\bibitem[\protect\citeauthoryear{Lahiri and Larsen}{Lahiri and
  Larsen}{2005}]{lahirilarsen2005}
Lahiri, P. and M.~Larsen (2005).
\newblock Regression analysis with linked data.
\newblock {\em Journal of the American Statistical Association\/}~{\em 100
  (469)}, 222--230.

\bibitem[\protect\citeauthoryear{McLeod, Heasman, and Forbes}{McLeod
  et~al.}{2011}]{ESSnet}
McLeod, P., D.~Heasman, and I.~Forbes (2011).
\newblock Simulated data for the on the job training.
\newblock Downloaded from http://www.cros-portal.eu/content/job-training (22nd
  October 2018).

\bibitem[\protect\citeauthoryear{Pfeffermann}{Pfeffermann}{2013}]{Pfe13}
Pfeffermann, D. (2013).
\newblock New important developments in small area estimation.
\newblock {\em Statistical Science\/}~{\em 28 (1)}, 40--68.

\bibitem[\protect\citeauthoryear{Prasad and Rao}{Prasad and Rao}{1990}]{Pra90}
Prasad, N. G.~N. and J.~N.~K. Rao (1990).
\newblock The estimation of the mean squared error of small area estimators.
\newblock {\em Journal of the American Statistical Association\/}~{\em 85
  (409)}, 163--171.

\bibitem[\protect\citeauthoryear{Rao}{Rao}{2003}]{Rao03}
Rao, J. N.~K. (2003).
\newblock {\em Small Area Estimation}.
\newblock New York: Wiley.

\bibitem[\protect\citeauthoryear{Rao and Molina}{Rao and
  Molina}{2015}]{RaoMol15}
Rao, J. N.~K. and I.~Molina (2015).
\newblock {\em Small Area Estimation\/} (2nd Edition ed.).
\newblock New York: Wiley.

\bibitem[\protect\citeauthoryear{Richardson and Welsh}{Richardson and
  Welsh}{1995}]{Ric95}
Richardson, A.~M. and A.~H. Welsh (1995).
\newblock Robust restricted maximum likelihood in mixed linear models.
\newblock {\em Biometrics\/}~{\em 51 (4)}, 1429--1439.

\bibitem[\protect\citeauthoryear{Robinson}{Robinson}{1991}]{Rob91}
Robinson, G.~K. (1991).
\newblock That blup is a good thing: The estimation of random effects.
\newblock {\em Statistical Science\/}~{\em 6 (1)}, 15--32.

\bibitem[\protect\citeauthoryear{Samart and Chambers}{Samart and
  Chambers}{2014}]{Samartchambers2014}
Samart, K. and R.~Chambers (2014).
\newblock Linear regression with nested errors using probability-linked data.
\newblock {\em Australian and New Zealand Journal of Statistics\/}~{\em 56
  (1)}, 27--46.

\bibitem[\protect\citeauthoryear{Scheuren and Winkler}{Scheuren and
  Winkler}{1993}]{Scheurenwinkler1993}
Scheuren, F. and W.~Winkler (1993).
\newblock Regression analysis of data files that are computer matched.
\newblock {\em Survey Methodology\/}~{\em 19}, 39--58.

\bibitem[\protect\citeauthoryear{Scheuren and Winkler}{Scheuren and
  Winkler}{1997}]{Scheurenwinkler1997}
Scheuren, F. and W.~Winkler (1997).
\newblock Regression analysis of data files that are computer matched - part
  ii.
\newblock {\em Survey Methodology\/}~{\em 23}, 157--165.

\bibitem[\protect\citeauthoryear{Sinha and Rao}{Sinha and Rao}{2009}]{Sin09}
Sinha, S.~K. and J.~N.~K. Rao (2009).
\newblock Robust small area estimation.
\newblock {\em The Canadian Journal of Statistics\/}~{\em 37 (3)}, 381--399.

\bibitem[\protect\citeauthoryear{Tzavidis, Marchetti, and Chambers}{Tzavidis
  et~al.}{2010}]{Tza10}
Tzavidis, N., S.~Marchetti, and R.~Chambers (2010).
\newblock Robust estimation of small area means and quantiles.
\newblock {\em Australian and New Zealand Journal of Statistics\/}~{\em 52
  (2)}, 167--186.

\bibitem[\protect\citeauthoryear{Winkler}{Winkler}{2009}]{Winkler2009}
Winkler, W.~E. (2009).
\newblock {\em Sample Surveys: Design, Methods and Applications}, Chapter
  Record Linkage, pp.\  351--380.
\newblock New York: Elsevier B.V.

\bibitem[\protect\citeauthoryear{Winkler}{Winkler}{2014}]{Winkler2014}
Winkler, W.~E. (2014).
\newblock Matching and record linkage.
\newblock {\em WIREs Computational Statistics\/}~{\em 6}, 313--325.

\end{thebibliography}

\end{document}